\documentclass[11pt, draftcls, onecolumn]{IEEEtran}
\usepackage{amsmath}
\usepackage{amsthm}
\usepackage{amsfonts}
\usepackage{amssymb}
\usepackage{amscd}
\usepackage{cite}
\usepackage[dvips]{graphicx} \def\BibTeX{{\rm B\kern-.05em{\sc i\kern-.025em b}\kern-.08em
    T\kern-.1667em\lower.7ex\hbox{E}\kern-.125emX}}
\newtheorem{theorem}{Theorem}
\usepackage{subfigure}
\usepackage{newlfont}
\usepackage{enumerate}
\usepackage{amssymb}
\usepackage{latexsym}
\usepackage{epsfig}
\usepackage{color}
\usepackage{mathtools}
\usepackage{algorithmic}
\usepackage{multicol}
\usepackage{etoolbox}
\usepackage{framed}
\newtheorem{lemma}{Lemma}
\newtheorem{proposition}{Proposition}
\newtheorem{corollary}{Corollary}
\newcounter{hints}

\renewcommand{\thehints}{\alph{hints}}
\newcommand{\hintedrel}[2][]{%
  \stepcounter{hints}%
  \if\relax\detokenize{#1}\relax\else\csxdef{hint@#1}{\thehints}\fi
  \mathrel{\overset{\textrm{(\thehints)}}{\vphantom{\le}{#2}}}%
}

\newcommand{\lb}{\left(}
\newcommand{\rb}{\right)}
\newcommand{\ls}{\left[}
\newcommand{\rs}{\right]}

\newcommand{\uy}{\underline{\mathbf{y}}}
\newcommand{\matX}{\mathbf{X}}

\newcommand{\matA}{\mathbf{A}}
\newcommand{\mc}[1]{\mathcal{#1}}

\newcommand{\mb}[1]{\mathbb{#1}}
\newcommand{\ul}[1]{\underline{#1}}

\long\def\symbolfootnote[#1]#2{\begingroup%
\def\thefootnote{\fnsymbol{footnote}}\footnote[#1]{#2}\endgroup}

\makeatletter
\renewcommand\subsubsection{\@startsection{subsubsection}{3}{0mm}{0ex plus 0.1ex minus 0.1ex}%
{0.3ex plus 0ex}{\normalfont\normalsize\itshape}}%
\makeatother

\title{Computationally Tractable Algorithms for Finding a Subset of Non-defective Items from a Large Population} 

\author{Abhay Sharma and Chandra R. Murthy\\
  Dept. of ECE,  Indian Institute of Science, Bangalore 560 012, India \\
  abhay.bits@gmail.com, cmurthy@ece.iisc.ernet.in
         }

\begin{document}
\maketitle
\vspace{-1cm}
\begin{abstract}
In the classical non-adaptive group testing setup,  pools of items are tested 
together, and the main goal of a recovery algorithm 
is to identify the \emph{complete defective set} given the outcomes of
different group tests. In contrast, the main goal of a \emph{non-defective subset recovery} 
algorithm is to identify a \emph{subset} of non-defective items given the test
outcomes.
In this paper, we present a suite of computationally efficient and analytically tractable 
non-defective subset recovery algorithms. 
By analyzing the probability of error of the algorithms, we obtain bounds on the number of tests required for  
non-defective subset recovery with arbitrarily small probability of error.
Our analysis accounts for the impact of both the additive noise 
(false positives) and dilution noise (false negatives).
By comparing with the information theoretic lower bounds, we show that the upper 
bounds on the number of tests are order-wise tight up to a $\log^2K$ factor, where 
$K$ is the number of defective items.
We also provide simulation results that compare the relative performance 
of the different algorithms and provide further insights into their practical utility. 
The proposed algorithms significantly outperform the straightforward approaches of testing items one-by-one, and of 
first identifying the defective set and then choosing the non-defective items from the complement set, 
in terms of the number of measurements required to ensure a given success rate.
\end{abstract}


\section{Introduction}
\label{sec:Introduction}
The general group testing framework\cite{Dorfman_GT, du2006} considers a large set of $N$ items, in which 
an unknown subset of $K$ items possess a certain testable property, e.g., the 
presence of an antigen in a blood sample, 
presence of a pollutant in an air sample, etc. This subset is referred to as 
the ``defective'' subset, and
its complement is referred to as the ``non-defective'' or ``healthy'' subset.
A defining notion of this framework is the \emph{group test}, a test that operates 
on a \emph{group} of items and provides a binary indication as to whether 
or not the property of interest is present collectively in the group.
A \emph{negative} indication implies that none of the tested items 
are defective. A \emph{positive} indication
implies that at least one of the items is defective. 
In practice, due to the hardware and test procedure limitations, the group tests are 
not completely reliable.
Using the outcomes of multiple such (noisy) group tests, a
basic goal of group testing is to reliably identify the defective set of items with as few
tests as possible. 
The framework of group testing has found applications in diverse
engineering fields such as industrial testing\cite{sobel1959group}, DNA sequencing\cite{ND00,du2006},
data pattern mining\cite{Macula_data_pattern,MV_ICML_2013,DashMV14}, 
medical screening\cite{du2006}, multi-access communications\cite{du2006,Wolf_multi_access}, 
data streaming\cite{Cormode_whatshot, Gilbert_GT}, etc.

One of the popular versions of the above theme is the non-adaptive 
group testing (NGT), where different tests are conducted simultaneously, i.e.,
the tests do not use information provided by the outcome of any other test. 
NGT is especially useful when the individual
tests are time consuming, and hence the testing time associated with 
adaptive, sequential testing is prohibitive.
An important aspect of NGT is how to determine the set of 
individuals that go into each group test. Two main approaches exist:
a combinatorial approach, see e.g., \cite{kautz_nonrand_supimp, Erdos_cover,Ruszinko94},
which considers explicit constructions of test matrices/pools; and 
a probabilistic pooling approach, see e.g., \cite{Dachkov_bound, Sebo1985, Gilbert_GT},
where the items included in the group test are chosen uniformly at random
from the population. Non-adaptive group testing has also been referred to as \emph{boolean} compressed sensing in the recent literature~\cite{Atia_BooleanCS, jaggi_gtalgo}.

In this work, in contrast to the defective set identification problem, we study the 
\emph{healthy/non-defective subset identification} problem in the
noisy, non-adaptive group testing with random pooling (NNGT-R) framework.
There are many applications where the goal is to identify
only a small subset of non-defective items. 
For example, consider the spectrum hole search problem in a cognitive 
radio (CR) network setup. It is known that 
the primary user occupancy is sparse in the frequency domain, over a wide band
of interest\cite{Cabric05cognitive,FCCreport}. 
This is equivalent to having a small subset of defective 
items embedded in a large set of candidate frequency bins.
The secondary users do not need to identify all the frequency bins occupied by the
primary users; they only need to discover a small number of unoccupied sub-bands to 
setup the secondary communications. This, in turn, is a non-defective subset 
identification problem when the bins to be tested for primary occupancy can be
pooled together into group tests~\cite{AC_tvt_14}.
In \cite{AC13}, using information theoretic arguments,
it was shown that compared to the conventional approach of identifying the non-defective 
subset by first identifying the defective set, 
directly searching for an $L$-sized non-defective subset offers a reduction in the number of tests, 
especially when $L$ is small compared to $N-K$.  
The achievability results in \cite{AC13} were obtained by analyzing the performance of the exhaustive
search based algorithms which are not practically implementable.
In this paper, we develop computationally efficient algorithms for non-defective subset 
identification in an NNGT-R framework.

We note that the problem of non-defective subset identification is a generalization of the defective set identification problem, in the sense that, when $L=N-K$, the non-defective subset identification problem is
identical to that of identifying the $K$ defective items.
Hence, by setting $L=N-K$, the algorithms presented in this work can be related to
algorithms for 
finding the defective set; see \cite{du2006} for an excellent collection of 
existing results and references.
In general, for the NNGT-R framework, three broad approaches have been adopted
for defective set recovery\cite{jaggi_gtalgo}.
First, the row based approach (also frequently referred to as the ``na\"{i}ve'' decoding algorithm)
finds the defective set by finding \emph{all} the 
non-defective items. The survey in \cite{chen2008survey} lists many variants of 
this algorithm for finding defective items. More recently, the {\bf CoCo} 
algorithm was studied in \cite{jaggi_gtalgo}, where an interesting connection 
of the na\"{i}ve decoding algorithm with the classical coupon-collector problem was established
for the noiseless case.
The second popular decoding approach is based on the idea of
finding defective items iteratively (or greedily) by matching 
the column of the test matrix corresponding to a given item with the test outcome 
vector\cite{du2006, vetterli_ngt, jaggi_gtalgo,Johnson_2014}. 
For example, in \cite{vetterli_ngt}, column matching consists of taking set differences 
between the set of pools where the item is tested and the set of pools with positive outcomes.
Another variant of matching is considered in \cite{jaggi_gtalgo}, where, for a given column, the ratio
of number of times an item is tested in pools with positive and negative outcomes is computed and
compared to a threshold.
A recent work, \cite{Yoo_arxiv_2013}, investigates the problem of finding zeros in 
a sparse vector in the compressive sensing framework, and also proposes
a greedy algorithm based on correlating the columns of the sensing matrix 
 (i.e., column matching) with the output
vector.\footnote{Note that directly computing correlations between column vector for an item and the test
outcome vector will not work in case of group testing, as both the vectors are boolean.
Furthermore, positive and negative pools have asymmetric roles in the group testing
problem.}
The connection between defective set identification in group testing and 
the sparse recovery in compressive sensing was further highlighted in \cite{mtov_mtov_2012,jaggi_gtalgo, DashMV14}, where  relaxation based linear programming algorithms have been proposed for defective 
set identification in group testing. 
A class of linear programs to solve the defective set identification problem 
was proposed by letting the boolean variables take real values (between $0$ and $1$) and
setting up inequality or equality constraints to model the outcome of each pool.

In this work, we develop novel algorithms for identifying a non-defective subset in an
NNGT-R framework.
We present error rate analysis for each algorithm and derive 
non-asymptotic upper bounds on the average error rate.
The derivation leads to a theoretical guarantee on the sample complexity, i.e., 
the number of tests required to identify a subset of non-defective items with arbitrarily small 
probability of error. 
We summarize our main contributions as follows:
\begin{itemize}
  \item We propose a suite of computationally efficient and analytically tractable algorithms
    for identifying a non-defective subset of given size in a NNGT-R framework: {\bf RoAl} (row based),  {\bf CoAl} (column based) and {\bf RoLpAl}, {\bf RoLpAl++}, {\bf CoLpAl} (Linear Program (LP) relaxation based) algorithms.
  \item We derive bounds on the number of tests that guarantee successful non-defective 
    subset recovery for each algorithm. The derived bounds are a function of
    the system parameters, namely, the number of defective items, the size of non-defective subset, 
    the population size, and the noise parameters. Further,
    \begin{itemize}
      \item The presented bounds on the number of tests for different algorithms are within
	$O(\log^2 K)$ factor, where $K$ is the number of defective items, of the information 
	theoretic lower bounds which were derived in our past work\cite{AC13}.
      \item For our suite of LP based algorithms, we present a novel 
	analysis technique based on 
	characterizing the recovery conditions via the dual variables associated with 
	the LP, which may be of interest in its own right.
\end{itemize}
  \item Finally, we present numerical simulations to compare the relative 
    performance of the algorithms. The results also illustrate the significant benefit in finding non-defective items directly, compared to using the existing defective set recovery methods or testing items one-by-one, in terms of the number of group tests required.
\end{itemize}

%
The rest of the paper is organized as follows.
Section~\ref{sec:PrbSetup} describes the NNGT-R framework and the problem setup.
The proposed algorithms and the main analytical results are presented 
in Section~\ref{result_alg}. The proofs of the main results are provided in
Section~\ref{sec_thm_proofs}.
Section~\ref{sec:simulations} discusses the numerical simulation results, and 
the conclusions are presented in Section~\ref{sec_conclusions}.
We conclude this section by presenting the notation followed throughout the paper.

\noindent \textbf{Notation:}
Matrices are denoted using uppercase bold letters 
and vectors are denoted using an underline.
For a given matrix $\matA$, $\ul{a}_i^{(r)}$ and $\ul{a}_i$ denote the $i^{\text{th}}$ 
row and column, respectively. 
For a given index set $S$, $\matA(S,:)$ denotes a sub-matrix of $\matA$ where only the 
rows indexed by set $S$ are considered. Similarly, $\matA(:,S)$ or $\matA_S$ denotes a 
sub-matrix of $\matA$ that consists only of columns indexed by set $S$. 
For a vector $\ul{a}$, $\ul{a}(i)$ denotes its $i^{\text{th}}$ component; 
$\text{supp}(\ul{a}) \triangleq \{j: \ul{a}(j) > 0 \}$;
$\{\ul{a} = c\}$ denotes the set $\{j: \ul{a}(j) = c \}$ for any $c$.
In the context of a boolean vector, $\ul{a}^c$ denotes the component wise 
boolean complement of $\ul{a}$.
$\ul{1}_{n}$ and $\ul{0}_{n}$ denote an all-one and all-zero vector, respectively, of
size $n \times 1$.
We denote the component wise inequality as $\ul{a} \preccurlyeq \ul{b}$, i.e., it means
$\ul{a}(i) \le \ul{b}(i)~\forall~i$.
Also, $\ul{a} \circ \ul{b}$ denotes the component-wise product,
i.e., $(\ul{a} \circ \ul{b})(i) = \ul{a}(i) \ul{b}(i), ~\forall~i$.
The boolean OR operation is denoted by ``$\bigvee$''.
For any $q \in [0, 1]$, $\mathcal{B}(q)$ denotes the Bernoulli distribution 
with parameter $q$.
$\mathbb{I}_{\mathcal{A}}$ denotes the indicator function and returns $1$ if
the event $\mathcal{A}$ is true, else returns $0$.
Note that, $x(n)=O(y(n))$ implies that $\exists~B>0$ and $n_0 > 0$, such that 
$|x(n)| \le B |y(n)|$ for all $n > n_0$. Further, 
$x(n)=\Omega(y(n))$ implies that $\exists~B>0$ and $n_0 > 0$, such that 
$|x(n)| \ge B |y(n)|$ for all $n > n_0$. Also, $x(n) = o(y(n))$ implies that for every $\epsilon > 0$, there exists an $n_0 > 0$ such that $|x(n)| \le \epsilon |y(n)|$ for all $n > n_0$. 
All logarithms in this papers are to the base $e$.
Also, for any $p \in [0, 1]$, $H_b(p)$ denotes the binary entropy in nats, i.e.,
$H_b(p) \triangleq -p \log(p) - (1-p) \log(1-p)$.

\section{Signal Model} \label{sec:PrbSetup}
In our setup, we have a population of $N$ items, out of which $K$ are defective.
Let $\mathcal{G} \subset [N]$ denote the
defective set, such that $| \mathcal{G} | = K$. 
We consider a non-adaptive group testing framework with random pooling~\cite{Atia_BooleanCS, du2006, malyutov_1, jaggi_gtalgo}, where the items to be pooled in a given test
are chosen at random from the population.
The group tests are defined
by a boolean matrix, $\matX \in \{0,1\}^{M \times N}$, that assigns
different items to the $M$ group tests (pools). 
The $j^{\text{th}}$ pool tests the items corresponding to the columns 
with $1$ in the $j^{\text{th}}$ row of $\matX$.
We consider an i.i.d.\ random Bernoulli
measurement matrix\cite{Atia_BooleanCS}, where each $X_{ij} \sim \mathcal{B}(p)$ for some
$0 < p < 1$. 
Thus, $M$ randomly generated pools are specified. In the above, $p$ is a 
design parameter that controls the average group size, 
i.e., the average number of items being tested in a single group test.
In particular, we choose $p=\frac{\alpha}{K}$, and a specific value of $\alpha$ 
is chosen based on the analysis of different algorithms.

If the tests are completely reliable, then
the output of the $M$ tests is given by the boolean OR of the 
columns of $\matX$ corresponding to the defective set $\mathcal{G}$.
However, in practice, the outcome of a group test may be unreliable. Two popular noise models that are 
  considered in the literature on group testing are\cite{vetterli_ngt,Atia_BooleanCS,jaggi_gtalgo}:
(a) An \emph{additive} noise model, where there is a 
probability, $q \in (0,0.5) $, that the outcome of a group test containing
only non-defective items turns out to be positive (Fig.~\ref{figure:noise_model});  
(b) A \emph{dilution} model, where there 
is a probability, $u \in (0,0.5)$, that a given item does not participate
in a given group test (see Fig.~\ref{figure:noise_model}).
Let $\underline{d}_i \in \{0,1\}^M$. Let 
$\underline{d}_i(j) \sim \mathcal{B}(1-u)$ be chosen independently 
for all $j=1, 2, \ldots, M$ and for all $i=1, 2, \ldots, N$. Let 
$\mathbf{D}_i \triangleq \text{diag}(\underline{d}_i)$.
The output vector $\underline{y} \in \{0, 1\}^M$ can be represented as
\begin{align} \label{eq:gtmodel}
  \underline{y} = \bigvee_{i=1}^{N} \mathbf{D}_i 
  \underline{x}_i\mathbb{I}_{\{i \in \mathcal{G}\}}
  \bigvee \underline{w},
\end{align}
where $\underline{x}_i \in \{0, 1\}^M$ is the $i^{\text{th}}$ column of $\matX$, 
$\underline{w} \in \{0, 1\}^M$ is the additive noise with the $i^{\text{th}}$ component
$\underline{w}(i) \sim \mathcal{B}(q)$.
Note that, for the noiseless case, $u=0, q=0$. 
Given the test output vector, $\underline{y}$, our goals are as follows:
\begin{enumerate}[(a)]
  \item To find computationally tractable algorithms to identify $L$ 
non-defective items, i.e., an $L$-sized subset belonging to 
$[N] \backslash \mathcal{G}$. 
 \item  To analyze the performance of the proposed algorithms
   with the objective of (i) finding the number of tests and 
   (ii) choosing the appropriate design parameters that leads 
   to non-defective subset recovery with high probability of success.
\end{enumerate}

\begin{figure}[t]
\centering
\includegraphics[scale=0.8]{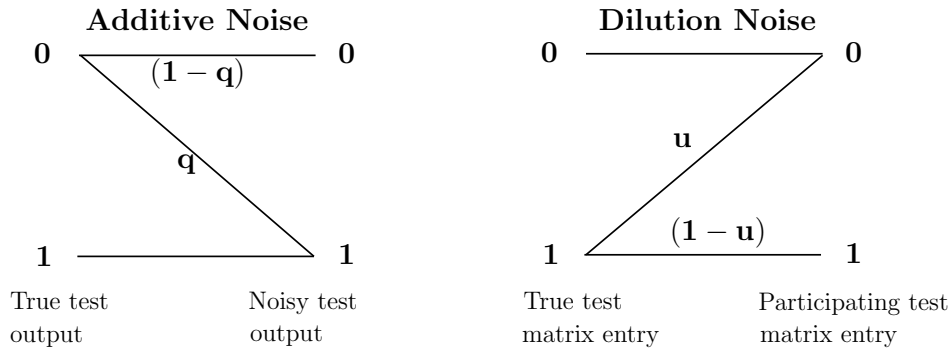}
\caption{Impact of different types of noise on the group testing signal model.}
\label{figure:noise_model}
\end{figure}

In the literature on defective set recovery
in group testing or on sparse vector recovery in compressed sensing,  
there exist two type of recovery results:
(a) \emph{Non-uniform/Per-Instance recovery results}:  
These state that a randomly chosen test matrix leads to non-defective subset
recovery with high probability of success for a given fixed defective set and,
(b) \emph{Uniform/Universal recovery results}: 
These state that a random draw of the test matrix leads to a successful non-defective subset
recovery with high probability for all possible defective sets.
It is possible to easily extend non-uniform results to the
uniform case using union bounds. Hence, we focus mainly on non-uniform recovery
results, and demonstrate the extension to the uniform case for
one of the proposed algorithms (see Corollary~\ref{cor_A1}).
%
Note that the non-uniform scenario is equivalent to the uniform recovery scenario
when the defective set is chosen uniformly at random from the set of ${N \choose K}$ 
possible choices. For the latter scenario,
information theoretic lower bounds on the number of tests  
for the non-defective subset recovery problem were derived in \cite{AC13} using Fano's inequality. We use these
bounds in assessing the performance of the proposed algorithms (see Section~\ref{sec_alg_disc}).
For the ease of reference, we summarize these results in Table~\ref{tab:tab_necc_order_1}.

For later use, we summarize some key facts pertaining to the above signal 
model in the lemma below. 
For any $l \in [M]$ and $k \in [N]$, let $X_{lk}$
denote the $(l,k)^{\text{th}}$ entry of the test matrix $\matX$
and let $Y_l \triangleq \ul{y}(l)$ denote the $l^{\text{th}}$ test output.
With $u$, $q$ and $p$ as defined above,  
let $\Gamma \triangleq (1-q) \lb 1 - (1-u)p \rb^{K}$ and 
$\gamma_0 \triangleq \frac{u}{(1 - (1-u)p)}$.
Then it follows that,
\begin{lemma} \label{sig_mod_facts}
  \begin{enumerate}[(a)]
    \item $\mb{P}( Y_l = 0) = \Gamma$.
    \item For any $j \notin S_d$, $\mb{P}(Y_l | X_{lj})  = \mb{P}(Y_l)$.
    \item For any $i \in S_d$, $\mb{P}( Y_l = 0 | X_{li} = 1) = \gamma_0 \Gamma$ and 
      $\mb{P}( Y_l = 0 | X_{li} = 0) = \frac{\Gamma}{1 - (1-u)p}$.
     Further, using Bayes rule, $\mb{P}( X_{li} = 1 | Y_l = 0 ) = p \gamma_0$.
    \item Given $Y_l$, $X_{li}$ is independent of $X_{lj}$ for
      any $i \in S_d$ and $j \notin S_d$.
  \end{enumerate}
\end{lemma}
The proof is provided in Appendix \ref{prf_sig_mod_facts}. 

\section{Algorithms and Main Results} \label{result_alg}
We now present several algorithms for non-defective/healthy subset recovery.
Each algorithm takes the observed noisy test-output vector $\ul{y} \in \{0,1\}^M$ and 
the test matrix $\matX \in \{0,1\}^{M \times N}$ as inputs, and 
outputs a set of $L$ items, $\hat{S}_L$, that have been declared non-defective.
The recovery is successful if the
declared set does not contain any defective item, i.e., $\hat{S}_L \cap S_d = \{ \emptyset \}$. 
For each algorithm, we derive expressions for the upper bounds on the 
average probability of error, which are 
further used in deriving the number of tests required for successful 
non-defective subset recovery. 

\subsection{Row Based Algorithm} \label{row_alg}
Our first algorithm to find non-defective items is also the simplest and the 
most intuitive one.
We make use of the basic fact of group testing that, in the noiseless case, if the 
test outcome is negative, then all the items being tested are non-defective.

\begin{framed}
\noindent {\bf{RoAl}} (Row based algorithm): 
\begin{itemize}
  \item Compute $\ul{z} =  \sum_{j \in \text{supp}(\ul{y}^c)} \ul{x}_j^{(r)}$, 
    where $\ul{x}_j^{(r)}$ is the $j^{\text{th}}$ row of the test matrix.
  \item Order entries of $\ul{z}$ in descending order.
  \item Declare the items indexed by the top $L$ entries as the non-defective subset.
  \end{itemize}
\end{framed}
That is, declare the $L$ items that have been tested most number of times in 
pools with negative outcomes as non-defective items.
The above decoding algorithm proceeds by only considering the tests with 
negative outcomes. 
Note that, when the test outcomes are noisy, there is a nonzero probability
of declaring a defective item as non-defective. In particular, the dilution
noise can lead to a test containing defective items in the pool being declared 
negative, leading to a possible misclassification of the defective items. 
On the other hand, since the algorithm only considers tests with negative outcomes, 
additive noise does not lead to misclassification of defective items 
as non-defective.  However, the additive noise does lead to an increased number of tests as
the algorithm has to possibly discard many of the pools that contain only
non-defective items.

We note that existing row based algorithms for finding defective set \cite{du2006, jaggi_gtalgo}
can be obtained as a special case of the above algorithm by setting $L=N-K$, i.e., by looking
for all non-defective items.
However, the analysis in the past work does not quantify the impact of the parameter $L$ and that
is our main goal here. We characterize the number of tests, $M$, that are 
required to find $L$ non-defective items with high probability of success using
{\bf RoAl} in Theorem~\ref{thm_A1A2}.

\subsection{Column Based Algorithm} \label{col_alg}
The column based algorithm is based on matching the columns of the test matrix with 
the test outcome vector. A non-defective item does not impact the output and hence
the corresponding column in the test matrix should be ``uncorrelated'' with the output. 
On the other hand, ``most'' of the pools that test a defective item should 
test positive. This forms the basis of distinguishing a defective item from 
a non-defective one. The specific algorithm is as follows:

\begin{framed}
\noindent {\bf{CoAl}} (Column based algorithm): 
Let $\psi_{cb} \ge 0$ be any constant. 
\begin{itemize}
  \item For each $i=1, \ldots, N$, compute
    \begin{align} \label{eq:greedy_ts}
      \mc{T}(i) = \underline{x}_i^T \ul{y}^c - \psi_{cb} (\underline{x}_i^T \underline{y}),
    \end{align}
    where 
    $\ul{x}_i$ is the $i^{\text{th}}$ column of $\matX$.
  \item Sort $\mc{T}(i)$ in descending order.
  \item Declare the items indexed by the top $L$ entries as the non-defective subset.
\end{itemize}
\end{framed}
We note that, in contrast to the row based algorithm, {\bf CoAl} works with pools of 
both the negative and positive test outcomes (when the parameter $\psi_{cb} > 0$; its choice is explained below). 
For both {\bf RoAl} and {\bf CoAl}, by analyzing the probability of error, we can derive
the sufficient number of tests required to achieve arbitrarily small error rates. 
We summarize the main result in the following theorem:
\begin{theorem} \label{thm_A1A2}
  (Non-Uniform recovery with {\bf {RoAl}} and {\bf CoAl}) 
  Let $\Gamma \triangleq (1-q) \lb 1 - (1-u)p \rb^{K}$ and 
  $\gamma_0 \triangleq \frac{u}{(1 - (1-u)p)}$.
  Suppose $K > 1$ and 
  let $p$ be chosen as $\frac{\alpha}{K}$ with $\alpha = \frac{1}{(1-u)}$.
  For {\bf RoAl}, let $\psi_0 \triangleq 0$.
  For {\bf CoAl}, choose $\psi_0 \triangleq \frac{\gamma_0\Gamma}{1-\gamma_0 \Gamma}$ and
  set $\psi_{cb} = \psi_0$.
  Let $c_0 > 0$ be any constant.
  Then, there exist absolute constants $C_{a1}, C_{a2} > 0$ independent of $N$, $L$ and $K$, and
  different for each algorithm, such that, if the number of tests is chosen as
  \begin{align} \label{eq:M_cond_noisy_rba1}
    M \ge  (1 + c_0) \frac{ K (1-u) }{(1-q)(1-\gamma_0)^2 (1 + \psi_0)} 
    \lb \frac{C_{a1} \log \left [K {N - K \choose L - 1} \right] }{(N-K) - (L-1)} + C_{a2} \log K \rb,
  \end{align}
  then, for a given defective set, the algorithms {\bf{RoAl}} and {\bf CoAl} find $L$ non-defective items 
  with probability exceeding $1-\exp\left(-c_0 \log \lb K {N - K \choose L - 1} \rb \right)$ $- \exp(-c_0 \log K)$.
\end{theorem}
The following corollary extends Theorem~\ref{thm_A1A2} to uniform recovery of a 
non-defective subset using {\bf RoAl} and {\bf CoAl}.
\begin{corollary} \label{cor_A1}
  (Uniform recovery with {\bf RoAl} and {\bf CoAl}) 
  For any positive constant $c_0 > 0$,  there exist absolute constants $C_{a1}, C_{a2} > 0$ independent 
  of $N$, $L$ and $K$, and different for each algorithm, such that if the number of tests is chosen as
  \begin{align} \label{eq:M_cond_noisy_rba2}
    M \ge  (1 + c_0) \frac{ K (1-u) }{(1-q)(1-\gamma_0)^2 (1 + \psi_0)} 
    \lb \frac{C_{a1}\log \left [K {N - K \choose L - 1} {N \choose K} \right] }{(N-K) - (L-1)} + C_{a2} \log N \rb,
  \end{align}
  then for {\bf any} defective set,
the algorithms {\bf{RoAl}} and {\bf CoAl} find $L$ non-defective items with probability exceeding 
  $1-\exp\left(-c_0 \log \lb K {N - K \choose L - 1} \rb \right)$ $- \exp(-c_0 \log N)$.
\end{corollary}
The proof of the above theorem is presented in Section~\ref{sec_prf_thm_A2}.
It is tempting to compare the performance of {\bf RoAl} and {\bf CoAl} by comparing
the required number of tests as presented in
(\ref{eq:M_cond_noisy_rba1}).
However, such comparisons must be done keeping in mind that the
required number of observations in (\ref{eq:M_cond_noisy_rba1}) are
based on an upper bound on the average probability of error. 
The main objective of these results is to provide a guarantee on the number of 
tests required for non-defective subset recovery and
highlight the order-wise dependence of the number of tests on the system parameters.
For the comparison of the relative performance of the algorithms, we refer the
reader to Section~\ref{sec:simulations}, where we present numerical results obtained
from simulations. From the simulations, we observe that {\bf CoAl} performs better than
{\bf RoAl} for most scenarios of interest. This is because, in contrast to {\bf RoAl}, {\bf CoAl} uses the information obtained from pools corresponding to both negative and positive test outcomes.


\subsection{Linear program relaxation based algorithms} \label{lin_alg}
In this section, we
consider linear program (LP) relaxations to the non-defective subset 
recovery problem and identify the conditions under which such LP relaxations lead
to recovery of a non-defective subset with high probability of success.
These algorithms are inspired by analogous algorithms studied in the context
of defective set recovery in the literature~\cite{mtov_mtov_2012,jaggi_gtalgo}. 
However, 
past analysis on the number of tests
for the defective set recovery do not carry over to the non-defective subset recovery 
because the goals of the algorithms are very different.
Let $Y_z \triangleq \{l \in [M]: \uy(l) = 0 \}$, i.e., $Y_z$ is
the index set of all the pools whose test outcomes are negative and $M_z \triangleq |Y_z|$. 
Similarly, let $Y_p \triangleq \{l \in [M]: \uy(l) = 1 \}$ and $M_p \triangleq |Y_p|$. 
Define the following linear program, with optimization variables 
$\ul{z} \in \mathbb{R}^{N}$ and $\ul{\eta}_z \in \mathbb{R}^{M_z}$:
\begin{align} \label{eq:row_lp_0}
  &  & \mathop{\text{minimize}}_{\ul{z}, \ul{\eta}_z} ~ ~ & ~ ~  \ul{1}_{M_z}^T \ul{\eta}_z \\
  \label{eq:row_lp_0_1}
  &\text{\bf (LP0)}  ~ ~ & \text{subject to} ~ ~ & ~ ~ 
  \matX(Y_z,:) (\ul{1}_{N} - \ul{z} ) - \ul{\eta}_z = \ul{0}_{M_z}, \\
  \nonumber
  & & & ~ ~ \ul{0}_{N} \preccurlyeq \ul{z} \preccurlyeq \ul{1}_{N} ,  ~ ~\ul{\eta}_z \succcurlyeq \ul{0}_{M_z}, \\
  \nonumber
  & & & ~ ~ \ul{1}_{N}^T \ul{z} \le L.
\end{align}

Consider the following algorithm:\footnote{The other algorithms 
presented in this sub-section, namely, {\bf RoLpAl++} and {\bf CoLpAl}, 
  have the same structure and differ only in the linear program being
  solved.}
\begin{framed}
\noindent{\bf RoLpAl} (LP relaxation with negative outcome pools only) 
\begin{itemize}
  \item Setup and solve {\bf LP0}. Let $\hat{\ul{z}}$ be the solution of {\bf LP0}.
  \item Sort $\hat{\ul{z}}$ in descending order.
  \item Declare the items indexed by the top $L$ entries as the non-defective subset.
\end{itemize}
\end{framed}

The above program relaxes the combinatorial problem of choosing $L$ out of $N$ items
by allowing the boolean
variables to acquire ``real'' values between $0$ and $1$ as long as 
the constraints imposed by negative pools, specified in (\ref{eq:row_lp_0_1}), are met.
Intuitively, the variable $\ul{z}$ (or the variable $[\ul{1}_N - \ul{z}])$ can be thought of as the 
confidence with which an item is being declared as non-defective (or defective). 
The constraint $\ul{1}_{N}^T \ul{z} \le L$ forces 
the program to assign high values (close to $1$) for ``approximately'' the top $L$ entries only,
which are then declared as non-defective.

For the purpose of analysis, we first derive sufficient conditions for correct non-defective subset 
recovery with {\bf RoLpAl} in terms of the dual variables of {\bf LP0}. We then derive the number
of tests required to satisfy these sufficiency conditions with high probability.
The following theorem summarizes the performance of the above algorithm:
\begin{theorem} \label{thm_A3}
  (Non-Uniform recovery with {\bf RoLpAl}) 
  Let $K > 1$ and let
  $p$ be chosen as $\frac{\alpha}{K}$ with $\alpha = \frac{1}{(1-u)}$.
  If the number of tests is chosen as in (\ref{eq:M_cond_noisy_rba1}) with $\psi_0 = 0$, 
  then for a given defective set there exist absolute constants $C_{a1}, C_{a2} > 0$ independent
  of $N$, $L$ and $K$, such that
  {\bf{RoLpAl}} finds $L$ non-defective items with probability 
  exceeding $1-\exp\left(-c_0 \log \lb K {N - K \choose L - 1}\rb \right)$ $- \exp(-c_0 \log K)$.
\end{theorem}
The proof of the above theorem is presented in Section~\ref{sec_prf_a3}.
Note that {\bf LP0} operates only on the set of pools with negative outcomes and is, thus,
sensitive to the dilution noise which can lead to a misclassification
of a defective item as non-defective. To combat this, we can leverage the information available
from the pools with 
positive outcomes also, by incorporating constraints for variables involved in these tests.
Consider the following linear program with optimization variables
$\ul{z} \in \mathbb{R}^{N}$ and $\ul{\eta}_z \in \mathbb{R}^{M_z}$:
\begin{align} \label{eq:row_lp_1prime}
&   &\mathop{\text{minimize}}_{\ul{z}, \ul{\eta}_z} ~ ~ & ~ ~  \ul{1}_{M_z}^T \ul{\eta}_z \\
  \nonumber
 &  \text{\bf (LP1)} ~ ~ &\text{subject to} ~ ~ & ~ ~ 
  \matX(Y_z,:) (\ul{1}_{N} - \ul{z} ) - \ul{\eta}_z = \ul{0}_{M_z} \\
  \label{eq:lp1_p_1}
 &  & & ~ ~ \matX(Y_p,:) (\ul{1}_{N} - \ul{z} ) \succcurlyeq (1 - \epsilon_0) \ul{1}_{M_p} \\
  \nonumber
&   & & ~ ~ \ul{0}_{N} \preccurlyeq \ul{z} \preccurlyeq \ul{1}_{N} ,  
  ~ ~\ul{\eta}_z \succcurlyeq \ul{0}_{M_z} \\
  \nonumber
 &  & & ~ ~ \ul{1}_{N}^T \ul{z} \le L.
\end{align}
In the above, $0 < \epsilon_0 \ll 1$ is a small positive constant. 
Note that (\ref{eq:lp1_p_1}) attempts to model,
in terms of real variables, a boolean statement that at least one 
of the items tested in tests with positive outcomes is a defective item. 
We refer to the algorithm based on {\bf LP1} as {\bf RoLpAl++}. 
We expect {\bf RoLpAl++} to outperform {\bf RoLpAl}, as  
the constraint (\ref{eq:lp1_p_1}) can provide further differentiation between
items that are indistinguishable just on the basis of negative pools.
Note that, due to the constraint $\ul{1}_{N}^T \ul{z} \le L$, the entries of  $\ul{\hat{z}}$ 
in $[N] \backslash \hat{S}_L$ are generally assigned small values. Hence,
when $L$ is small, for many of the positive pools, the
constraint (\ref{eq:lp1_p_1}) may not be active. Thus,
we expect {\bf RoLpAl++} to perform better than {\bf RoLpAl} as the value of
$L$ increases; this will be confirmed via simulation results in Section~\ref{sec:simulations}.
Due to the difficulty in obtaining estimates for the dual variables 
associated with the constraints (\ref{eq:lp1_p_1}), it is difficult to derive theoretical guarantees for {\bf RoLpAl++}. 
However, we expect the guarantees for {\bf RoLpAl++} to be similar to
{\bf RoLpAl}, and we refer the reader to Appendix~\ref{sec_prf_a3pp} for a discussion
regarding the same.


Motivated by the connection between {\bf RoAl} and {\bf RoLpAl}, as revealed in the
proof of Theorem~\ref{thm_A3} (see Section~\ref{sec_prf_a3}), we now propose another 
LP based non-defective subset recovery algorithm that incorporates both 
positive and negative pools, which, in contrast to {\bf RoLpAl++}, turns out 
to be analytically tractable.
By incorporating (\ref{eq:lp1_p_1}) in an unconstrained form and by using the \emph{same}
weights for all the associated Lagrangian multipliers in the optimization function,
we get:
\begin{align} \label{eq:row_lp_1}
  & & \mathop{\text{minimize}}_{\ul{z}} ~ ~ & ~ ~  \ul{1}_{M_z}^T \matX(Y_z,:) (\ul{1}_{N} - \ul{z} )
  - \psi_{lp} \left [\ul{1}_{M_p}^T \matX(Y_p,:) (\ul{1}_{N} - \ul{z} ) \right ] \\
  \nonumber
  & \text{\bf (LP2)} ~ ~& \text{subject to} ~ ~ & ~ ~ 
  \ul{0}_{N} \preccurlyeq \ul{z} \preccurlyeq \ul{1}_{N}, \\
  \nonumber
  & & &  ~ ~ \ul{1}_{N}^T \ul{z} \le L,
\end{align}
where $\psi_{lp} > 0$ is a positive constant that provides appropriate weights to the
two different type of cumulative errors.
Note that, compared to {\bf LP1}, we have also eliminated the equality constraints in the
above program.
The basic intuition is that by using (\ref{eq:lp1_p_1}) in an unconstrained form, i.e., by maximizing
$\sum_{j \in Y_z} \matX(j,:) (\ul{1}_{N} - \ul{z} )$, the 
program will tend to assign higher values to $(1 - \hat{\ul{z}}(i))$ (and hence lower
values to $\hat{\ul{z}}(i))$ for $i \in S_d$
since for random test matrices with i.i.d.\ entries, the defective 
items are likely to be tested more number of times in the pools with positive outcomes.
Also, in contrast to {\bf LP1} where different weightage is given to 
each positive pool via the value of the associated dual variable, {\bf LP2} gives the same weightage to each
positive pool, but it adjusts the overall weightage of positive pools using the constant
$\psi_{lp}$.
We refer to the algorithm based on {\bf LP2} as {\bf CoLpAl}. 
The theoretical analysis for {\bf CoLpAl} follows on similar lines as {\bf RoLpAl} and 
we summarize the main result in the following theorem:
\begin{theorem} \label{thm_A4}
  (Non-Uniform recovery with {\bf CoLpAl}) 
  Let $\Gamma \triangleq (1-q) \lb 1 - (1-u)p \rb^{K}$ and 
  $\gamma_0 \triangleq \frac{u}{(1 - (1-u)p)}$.
  Let $K > 1$ and let
  $p$ be chosen as $\frac{\alpha}{K}$ with $\alpha = \frac{1}{(1-u)}$.
  Let $\psi_0' \triangleq \min \left( \frac{\gamma_0 \Gamma} {1 - \gamma_0 \Gamma}, \frac{\Gamma}{2 (1 - \Gamma)}\right )$
  and set $\psi_{lp} = \psi_0'$.
  Then, for any positive constant $c_0$, there exist absolute constants $C_{a1}, C_{a2} > 0$ independent of $N$, $L$ and $K$,
  such that, if the number of tests is chosen as in (\ref{eq:M_cond_noisy_rba1}) with $\psi_0 = 0$,
  then for a given defective set {\bf{CoLpAl}} finds $L$ non-defective items with probability 
  exceeding ~$1-2\exp\left(-c_0 \log \lb K {N - K \choose L - 1}\rb \right) - \exp\left(-c_0 \log K \right)$.
\end{theorem}
An outline of the proof of the above theorem is presented in Section~\ref{sec_prf_thmA4}.

\section{Discussion on the Theoretical Guarantees} \label{sec_alg_disc}
We now present some interesting insights by analyzing 
the number of tests required for correct non-defective subset identification by the proposed 
recovery algorithms. 
We note that the expression in (\ref{eq:M_cond_noisy_rba1}) adapted for different algorithms
differs only on account of the constants involved. This allows us to present a unified analysis 
for all the algorithms.
\begin{enumerate}[(a)]
\item Asymptotic analysis of $M$ as $N \rightarrow \infty$: 
We consider the parameter regimes where $K, L \rightarrow \infty$ as $N \rightarrow \infty$.
We note that, under these regimes, when the conditions specified in the theorems are satisfied, the probability of decoding error can be made 
arbitrarily close to zero. In particular, we consider the regime where {$\frac{K}{N} \rightarrow \beta_0$, $\frac{L}{N} \rightarrow \alpha_0$, as $N \rightarrow \infty$, where $0 \le \beta_0 < \alpha_0 < 1$}, $\alpha_0 + \beta_0 < 1$. 
Define $\zeta \triangleq \frac{L-1}{N-K}$, and  $\zeta \rightarrow \zeta_0 \triangleq \frac{\alpha_0}{1 - \beta_0}$ as $N \rightarrow \infty$. Also, note that $\gamma_0 \rightarrow u$ as $N \rightarrow \infty$.
Using Stirling's formula, it can be shown that
$\lim_{N \rightarrow \infty} \frac{\log{N-K \choose L-1}}{(N- K)- (L - 1)} \le 
\frac{H_b(\zeta_0)}{1 - \zeta_0}$ (see, \cite{AC13}), 
where $H_b(\cdot)$ is the binary entropy function.
Further, let $g(\zeta) \triangleq \frac{H_b(\zeta)}{1 - \zeta}$.
Now, since  $g(\zeta_0)$ is a constant, the sufficient number of tests $M$ for the proposed algorithms depends on $K$ as
    $M  \ge C_0 \frac{K}{(1-u)(1-q)} \left( C_{a1} g(\zeta_0) + C_{a2} \log K + o(1) \right)$. Here, $C_0, C_{a1}$ and $C_{a2}$ are constants independent of $N, K, L, u$ and~$q$.

We compare the above with the sufficient number of test required for the defective set recovery
algorithms. 
When $K$ grows sub-linearly with $N$ (i.e., $\beta_0 = 0$), the sufficient number of tests for the proposed decoding algorithms
is $O(K \log K)$, which is better than the sufficient number of tests for finding the defective set, which scales as $O(K \log N)$\cite{vetterli_ngt,jaggi_gtalgo}. Whereas, for the 
regime where $K$ grows linearly with $N$ (i.e., $\beta_0 > 0$), the performance of the proposed algorithms is order-wise equivalent
to defective set recovery algorithms. 

We also compare the uniform recovery results.
The sufficient number of tests for uniform recovery as given in Corollary~\ref{cor_A1} 
for the algorithm {\bf RoAl} and {\bf CoAl} is $M = O(K \log N)$, which is significantly
better than the defective set recovery algorithms, where the sufficient number of tests 
scale as $O(K^2 \log (\frac{N}{K}))$~\cite{vetterli_ngt}.

  \item Variation of $M$ with $L$: 
    Let $\zeta$ and $g(\zeta)$ be as defined above. 
    We note that the parameter $L$ impacts $M$ \emph{only} 
    via the function $g(\zeta)$. Lemma~\ref{lemma_Gamma_approx} in Appendix~\ref{app_gamma_approx} shows that for small values (or even moderately high values) of $\zeta$, $g(\zeta)$ is upper bounded by an affine function in $\zeta$. This, in turn, shows that the sufficient number of tests is also approximately affine in $L$; this is also confirmed via simulation results in Section~\ref{sec:simulations}.

  \item Comparison with the information theoretic lower bounds: 
    We compare with the lower bounds on the number of tests for non-defective subset 
    recovery, as tabulated in Table~\ref{tab:tab_necc_order_1}. 
    For the noiseless case, i.e., $u=0, q=0$, the sufficient number of tests are
    within $O(\log^2 K)$ factor of the lower bound.
    For the additive noise only case, the proposed algorithms incur a factor of 
    $1/(1-q)$ increase in $M$. 
    In contrast, the lower bounds indicate that the number of tests is insensitive to additive noise, 
    when $q$ is close to $0$ (in particular, when $q < 1/K$).
    For the dilution noise case, the algorithms incur a factor $\frac{1}{(1-u)}$ increase in $M$, which is the same as in 
    the lower bound.
    We have also compared the number of tests obtained via simulations with an exact computation 
    of the lower bounds, and, interestingly, the algorithms fall within $O(\log K)$ factor of the
    lower bounds; we refer the reader to 
    Figure~\ref{figure:M_vs_L_lbnd_comp}, Section~\ref{sec:simulations}.


  \item Defective set recovery via non-defective subset recovery: It is interesting to note that by 
    substituting $L=N-K$ in (\ref{eq:M_cond_noisy_rba1}),
    we get $M = O \left( \frac{K \log(N-K)}{(1-u)(1-q)} \right)$, which is order-wise the
    same as the number of tests required for defective set 
    identification derived in the existing literature~\cite{vetterli_ngt,saligrama_CS,jaggi_gtalgo}. 

  \item Robustness under uncertainty in the knowledge of $K$: 
    The theoretical guarantees presented in the above theorems hold provided
    the design parameter $p$ is chosen as $O(\frac{1}{(1-u) K})$. This requires the
    knowledge of $u$ and $K$. {Note that the implementation of the recovery algorithms 
    do not require us to know the values of $K$ or $u$. These system model parameters are
    only required to choose the value of $p$ for constructing the test matrix.}
If $u$ and $K$ are unknown, similar guarantees can be derived, with a penalty on the number of tests. 
    For example, choosing $p$ as $O(1/K)$, i.e., independent of $u$, results in a 
    $\frac{1}{1-u}$ times increase in the number of tests.
    The impact of using an imperfect value of $K$ can also be quantified. Let
    $\hat{K}$ be the value used to design the test matrix and let $\Delta_k > 0$ be such that $\hat{K} = \Delta_k K$. 
    That is, $\Delta_k$ parametrizes the estimation error in $K$.
    Using the fact that for large $n$, $(1 - \alpha/n)^n \approx \exp(-\alpha)$, it follows that
    with $p = O(\frac{1}{\Delta_k K})$, the number of tests increases approximately 
    by a factor of 
    $f_M(\Delta_k) \triangleq \Delta_k \exp\left (-(1-u) \left( \frac{1}{\Delta_k} - 1 \right) \right )$ 
    compared to the case with perfect knowledge of $K$, i.e., with $p = O(1/K)$. 
    It follows that the proposed algorithms are robust to the uncertainty 
  in the knowledge of $K$.
    For example, with $u=0$, $f_M(1.5) = 1.09$, i.e., a $50\%$
    error in the estimation of $K$ leads to only a $9\%$ increase in the number of tests.
    Furthermore, the asymmetric nature of $f_M(\Delta_k)$ (e.g.,  $f_M(1.5) = 1.09$ and  $f_M(0.5) = 1.3$) 
    suggests that the algorithms are more robust when $\Delta_k > 1$ as compared to
    the case when $\Delta_k < 1$. We corroborate this behavior via numerical simulations 
    also (see Table~\ref{tab:robust_K}).

  \item Operational complexity: 
    The execution of {\bf RoAl} and {\bf CoAl} requires $O(M N)$ operations, where $M$ is the number of tests.
    The complexity of the LP based algorithms {\bf RoLpAl},
    {\bf RoLpAl++} and {\bf CoLpAl} are implementation dependent, but are, in general,
    much higher than  {\bf RoAl} and  {\bf CoAl}. 
    For example, an interior-point method based implementation will require
     $O(N^2(M+N)^{3/2})$ operations\cite{nesterov_lect}. Although this is higher than that of 
     {\bf RoAl} and {\bf CoAl}, it is still attractive in comparison to the brute force 
     search based maximum likelihood methods, due to its polynomial-time complexity.


    \end{enumerate}


\begin{table}[t]
  \centering
\caption{Finding a subset of $L$ non-defective items: Order results for necessary number 
of group tests which hold asymptotically as $N \rightarrow \infty$,
$\frac{K}{N} \rightarrow \beta_0$, $\frac{L}{N} \rightarrow \alpha_0$ and $\alpha_0 + \beta_0 < 1$
(see Theorem $3$, \cite{AC13}).}
  \begin{tabular}{|l|l|l|} \hline 
     No Noise ($u=0, q=0$)    & $\Omega \left (\frac{K}{\log K}\log \frac{1 - \beta_0}{1 - \alpha_0 - \beta_0} \right ) $\\ \hline
     Dilution Noise ($u>0, q=0$)
                 & $\Omega \left (\frac{K}{(1-u)\log K }\log \frac{1 - \beta_0}{1 - \alpha_0 - \beta_0} \right ) $ \\ \hline
    Additive Noise ($u=0, q > 0$)
       & $\Omega \left (\frac{K}{\min\left\{ \log \frac{1}{q}, \log K \right\} }\log \frac{1 - \beta_0}{1 - \alpha_0 - \beta_0} \right ) $ \\ \hline
  \end{tabular}
  \label{tab:tab_necc_order_1}
\end{table}

\section{Proofs of the Main Results} \label{sec_thm_proofs}
We begin by defining some quantities and terminology that is common to all the proofs.
In the following, we denote the defective set by $S_d$, such that $S_d \subset [N]$ and $|S_d| = K$.
We denote the set of $L$ non-defective items output by the decoding algorithm by 
$\hat{S}_L$.
For a given defective set $S_d$, 
$\mc{E} \triangleq \left \{ \hat{S}_L \cap S_d \neq \{\emptyset\} \right \}$ denotes
the error event, i.e., the event that a given decoding algorithm outputs an incorrect 
non-defective subset
and let $\text{Pr}(\mc{E})$ denote its probability.
Define $N_0 \triangleq (N-K) - (L-1)$.
We further let $S_z \subset [N]\backslash S_d$ denote any set of non-defective items 
such that $|S_z| = N_0$.  Also, we let $\mc{S}_z$ denote all such sets possible. 
Note that $|\mc{S}_z| = {N-K \choose L-1}$. 
Finally, recall from Lemma~\ref{sig_mod_facts} (Section~\ref{sec:PrbSetup}), $\Gamma \triangleq (1-q) \lb 1 - (1-u)p \rb^{K}$ and 
$\gamma_0 \triangleq \frac{u}{(1 - (1-u)p)}$.

\subsection{Proof of Theorem~\ref{thm_A1A2} and Corollary~\ref{cor_A1}} \label{sec_prf_thm_A2}

The proof involves upper bounding the probability of non-defective subset recovery 
error of the decoding algorithms, {\bf RoAl} and {\bf CoAl}, and identifying the parameter regimes where they can be made sufficiently small.

For {\bf CoAl}, recall that we compute the metric 
$\mc{T}(i) \triangleq \ul{x}_i^T {\ul{y}^c} - (\psi_{cb}) \ul{x}_i^T \ul{y}$ for
each item $i$ and output the set of items with the $L$ largest metrics as the non-defective set.
Clearly, for any item $i \in S_d$, 
if $i \in \hat{S}_L$, then there exists a set $S_z$ of non-defective items such that for 
all items $j \in S_z$, $\mc{T}(j) \le \mc{T}(i)$.
Thus, for {\bf CoAl}, it follows that,
  \begin{align} \label{eq:row_errevnt}
    \mc{E} \subset \mathop{\cup}_{i \in S_d} \{ i \in \hat{S}_L \} \subset
    \mathop{\cup}_{i \in S_d} \mathop{\cup}_{S_z \in \mc{S}_z} 
    \left[\mathop{\cap}_{j \in S_z} \{ \mc{T}(j) \le \mc{T}(i)\}\right] .
  \end{align}
  
The algorithm {\bf RoAl} succeeds when there exists a set of at 
least $L$ non-defective items that have been tested more number of times than 
any of the defective items, in the tests with negative outcomes. 
The number of times an item $i$ is tested in tests with negative outcomes 
is given by $\ul{z}(i)(=\ul{x}_i^T {\ul{y}^c})$, which is computed by {\bf RoAl}.
Hence, for any item $i \in S_d$, if $i \in \hat{S}_L$, then there 
exists a set $S_z$ of non-defective items such that for all items $j \in S_z$, 
$\ul{z}(j) \le \ul{z}(i)$.
And, thus, (\ref{eq:row_errevnt}) applies for {\bf RoAl} also, except with 
$\mc{T}$ replaced with $\ul{z}$.
Also, note that $\ul{z}(i) = \mc{T}(i) |_{\psi_{cb} = 0}$. This allows us
to unify the subsequent steps in the proof for the two algorithms.
We first work with the quantity $\mc{T}(i)$ and later  specialize the results for each algorithm.
The overall intuition for the proof is as follows:
For any $i$, since $\mc{T}(i)$ is a sum of independent random variables, it will tend 
to concentrate around its mean value.
For any $i \in S_d$ and $j \notin S_d$, we will show that the mean value 
of $\mc{T}(j)$ is larger than that of $\mc{T}(i)$. Thus, we expect the probability of the error event defined in (\ref{eq:row_errevnt}) to be small.

For any $i \in S_d$ and for any $j \notin S_d$, define 
$\mu_i \triangleq \mb{E}(\mc{T}(i))$, $\mu_j \triangleq \mb{E}(\mc{T}(j))$,
$\sigma_i^2 \triangleq \text{Var}(\mc{T}(i))$ and $\sigma_j^2 \triangleq \text{Var}(\mc{T}(j))$.
It follows that, 
\begin{align} \label{eq:mean_ij}
  \mu_j &= M p \lb \Gamma - \psi_{cb}(1 -\Gamma) \rb ~ ~ \mbox{and} ~ ~
  \mu_i = M p \lb \gamma_0 \Gamma - \psi_{cb}(1 -\gamma_0 \Gamma) \rb\\
  \label{eq:var_j}
  \sigma_j^2 &
  \le M p \lb \Gamma + \psi_{cb}^2(1 -\Gamma) \rb ~ ~ \mbox{and} ~ ~
  \sigma_i^2 
  \le M p \lb \gamma_0 \Gamma + \psi_{cb}^2(1 -\gamma_0 \Gamma) \rb.
\end{align}
An brief explanation of the above equations in presented in Appendix~\ref{prf_Z0Z1_stats}.
We note that, $(\mu_j - \mu_i) = M p \Gamma (1-\gamma_0) (1+\psi_{cb}) > 0$. 
To simplify (\ref{eq:row_errevnt}) further, we present the following proposition:
\begin{proposition} \label{prop_main_pe_ub}
Define $\tau \triangleq \frac{(\mu_j + \mu_i)}{2}$. Then, for any 
$\epsilon_0 > 0$ it follows that
  \begin{align} \label{eq:main_pe_rbcb_ub}
    \mb{\text{Pr}}(\mc{E}) &\le K {N-K \choose L -1} \lb P_{eh} \rb^{N_0} + K P_{ed}, 
  \end{align}
  where, $P_{eh} \triangleq \mb{P}\lb\{\mc{T}(j) < \tau + \epsilon_0\}\rb$ for any 
   $j \in S_z$ and  
   $P_{ed} \triangleq \mb{P}\lb \{\mc{T}(i) > \tau\} \rb$ for any $i \in S_d$.
\end{proposition}
The proof of the above proposition is presented in Section \ref{prf_prop_main_pe_ub}. Note that the above definitions of $P_{eh}$ and $P_{ed}$ are unambiguous because the corresponding probabilities are independent of the specific choice of indices $j$ and $i$, respectively.

Our next task is to bound $P_{eh}$ and $P_{ed}$ as defined in the above proposition.
For any $k$, since $\mc{T}(k)$ is a sum of $M$ independent random variables, each bounded by
$\max(1, \psi_{cb})$, we can use Bernstein's inequality~\cite{Lugosi06com}\footnote{For ease of 
reference, we have stated it in Appendix~\ref{sec_chernoff}.}
to bound the probability of their deviation from their mean values.
Since $\psi_{cb}$ is a free parameter, we proceed by assuming that $\psi_{cb} < 1$.
Thus, for any $i \in S_d$, with $\delta_0 \triangleq \tau - \mu_i = \frac{\mu_j - \mu_i}{2}$,
\begin{align} \label{eq:ped_ub}
  P_{ed} = P(\mc{T}(i) > \tau) = P(\mc{T}(i) > \mu_i + \delta_0) \le 
  \exp \lb - \frac{\delta_0^2}{2 \sigma_i^2 + \frac{2}{3} \delta_0 } \rb.
\end{align}
Similarly, for any $j \in S_z$, we choose $\epsilon_0 = \frac{\mu_j - \tau}{2} = \frac{\mu_j - \mu_i}{4}$, and get
\begin{align} \label{eq:peh_ub}
  P_{eh} = P(\mc{T}(j) < \tau + \epsilon_0) = P(\mc{T}(j) < \mu_j - \epsilon_0) \le 
  \exp \lb - \frac{\epsilon_0^2}{2 \sigma_j^2 + \frac{2}{3} \epsilon_0 } \rb.
\end{align}

We now proceed separately for each algorithm to arrive at the final results. 
Before that, we note that by choosing $p = \frac{\alpha}{K}$ with $\alpha = \frac{1}{(1-u)}$,
$\left[1 - \frac{(1-u) \alpha}{K} \right ]^K \ge \exp \left(- 2 \alpha (1-u) \right ) = e^{-2}$. 
This follows from the fact that for $0 < b < 1$, $(1 - b) \le e^{-b} \le 1 - \frac{b}{2}$. 
Thus, $(1-q) e^{-1} \ge \Gamma \ge  (1-q)e^{-2}$.
We also note that $\gamma_0 < 1$ for any $u < 0.5$ and for all 
$K > 1$.
\subsubsection{Proof for {\bf RoAl}}
For {\bf RoAl}, $\psi_{cb} = 0$. Thus, from (\ref{eq:mean_ij}) and (\ref{eq:var_j}) 
we have, $\mu_j - \mu_i = M p \Gamma (1-\gamma_0)$,
$\sigma_j^2 \le M p \Gamma$ and $\sigma_i^2 \le M p \gamma_0 \Gamma$.
Recall, $\delta_0 = \frac{\mu_j - \mu_i}{2}$ and $\epsilon_0 = \frac{\mu_j - \mu_i}{4}$.
Note that, $2 \sigma_i^2 + (2/3) \delta_0 < M p \Gamma \lb 2 \gamma_0 + (1-\gamma_0)/3 \rb 
< 2 M p \Gamma$.
Similarly, $2 \sigma_j^2 + (2/3) \epsilon_0 < M p \Gamma \lb 2 + (1-\gamma_0)/6 \rb < 3 M p \Gamma$. 
Thus, from (\ref{eq:ped_ub}) and (\ref{eq:peh_ub}), we have
\begin{align} \label{eq:peh_ub_1}
  P_{ed} \le \exp \lb - \frac{M p \Gamma (1 - \gamma_0)^2}{8}\rb ~ \mbox{and} ~ 
  P_{eh} \le \exp \lb - \frac{M p \Gamma (1 - \gamma_0)^2}{48}\rb.
\end{align}
Thus, choosing $p = \frac{1}{(1-u)K}$ and noting that $\Gamma \ge e^{-2} (1-q)$, from (\ref{eq:main_pe_rbcb_ub}) we get,
\begin{align} 
\nonumber
  \mb{P}(\mc{E}) \le \exp \ls - \frac{M (1-\gamma_0)^2(1-q) N_0}{C_{a1} K (1-u)}  +  
  { \log \lb K {N - K \choose L - 1} \rb} \rs
  + \exp \ls -\frac{M (1-\gamma_0)^2(1-q) }{C_{a2} K (1-u)}  + 
  \log K \rs, 
\end{align}
with $C_{a1} = 48 e^2$ and $C_{a2} = 8 e^2$. 
Thus, if $M$ is chosen as specified in (\ref{eq:M_cond_noisy_rba1}), 
with the constants $C_{a1}$, $C_{a2}$ chosen as above, then
the error probability is upper bounded by $\exp\left(-c_0 \log \left [K {N-K \choose L - 1}\right]  \right) 
+ \exp(-c_0 \log K)$. 

\subsubsection{Proof for \bf{CoAl}}
We first bound $P_{ed}$. 
With $\psi_{cb} = \psi_0$, where $\psi_0 \triangleq \frac{\gamma_0\Gamma}{1 - \gamma_0\Gamma}$,  we have 
$\sigma_i^2 \le M p \gamma_0\Gamma (1 + \psi_0)$. Also, we note that $\psi_0 < 1$.
Thus, $2 \sigma_i^2 + (2/3) \delta_0 < M p \Gamma (1 + \psi_0)\lb 2\gamma_0 + (1-\gamma_0)/3\rb 
< 2 M p \Gamma(1 + \psi_0)$. 
Thus, from (\ref{eq:ped_ub}), we get
  \begin{align}
    P_{ed} \le \exp \lb - \frac{M p \Gamma (1+\psi_0) (1 - \gamma_0)^2}{8}\rb.
  \end{align}
With $\psi_0$ as above, it follows that,
$2 \sigma_j^2 + (2/3) \epsilon_0 < M p \Gamma \lb 2 +  2 \frac{\gamma_0^2\Gamma}{(1-\gamma_0 \Gamma)} + \frac{(1-\gamma_0)(1+\psi_0)}{6} \rb
< 3 M p \Gamma (1+ \psi_0)$, since $1+\psi_0=\frac{1}{1 - \gamma_0 \Gamma}$.
Thus, from (\ref{eq:peh_ub}), we get
  \begin{align}
    P_{eh} \le \exp \lb - \frac{M p \Gamma (1+\psi_0) (1 - \gamma_0)^2}{48}\rb.
  \end{align}

The next steps follow exactly as for {\bf RoAl} and if $M$ is chosen as specified in (\ref{eq:M_cond_noisy_rba1}), 
with the constant $C_{a1}$ and $C_{a2}$ chosen as $48e^2$ and $8e^2$, respectively, then
the error probability remains smaller than $\exp\left(-c_0 \log \left [K {N-K \choose L - 1}\right]  \right) 
+ \exp(-c_0 \log K)$.

\subsubsection{Proof of Proposition~\ref{prop_main_pe_ub}} \label{prf_prop_main_pe_ub}
  For $i \in S_d$, define $\mc{H}_i \triangleq \{\mc{T}(i) \le \tau\}$.
  The error event in (\ref{eq:row_errevnt}) is a subset of the right hand side in the following equation:
  \begin{align} \label{eq:errevnt_row_col}
    \mc{E} \subset \mathop{\cup}_{i \in S_d}\left (  \{ \mathop{\cup}_{S_z \in \mc{S}_z} 
    \mathop{\cap}_{j \in S_z} (\mc{E}_{ij} \cap {\mc{H}_i}) \} \cup \overline{\mc{H}_i} \right ),
  \end{align}
  where $\mc{E}_{ij} \triangleq \{ \mc{T}(j) \le \mc{T}(i) \}$ for any $i \in S_d$ and $j \in S_z$.
  In the above, we have used the fact that, for any two sets $A$ and $B$, 
  $A \subset \{A \cap {B}\} \cup \overline{B}$. Further, using monotonicity properties, we have
  \begin{align} \label{eq:errevnt_sz}
    \left \{ \mc{E}_{ij} \cap {\mc{H}_i} \right \} \subset \{\mc{T}(j) \le \tau \} \subset \{ \mc{T}(j) < \tau + \epsilon_0\},
  \end{align}
  where $\epsilon_0 > 0$ is any constant.
  Consider any non-defective item $j\in S_z$. We note that for any given $\ul{y}$, $\mc{T}(j)$ can be 
  represented as a function of only $\ul{x}_j$, i.e., the $j^{\text{th}}$ column of the test matrix $\matX$. 
  From (\ref{eq:gtmodel}),  since $j \notin S_d$, the output is independent of the entries 
  of $\ul{x}_j$. Hence, for all $j \notin S_d$, and hence for all $j \in S_z$,  $\mc{T}(j)$'s are independent. 
  Using this observation, the claim in the proposition now follows from (\ref{eq:errevnt_row_col}) and
  (\ref{eq:errevnt_sz}) by accounting for the cardinalities of different sets involved in the union bounding.

\subsubsection{Proof of Corollary~\ref{cor_A1}}
For the uniform case, we use the union bound over all possible choices of the
defective set. 
The proof of the corollary follows same steps as the proof of Theorem~\ref{thm_A1A2}; the only
difference comes on account of the additional union bounding that has to be done to
account for all possible choices of the defective set.
Here, we briefly discuss the different multiplicative factors that have to be included
because of this additional union bound.
Let $\mc{S}_d$ denote the set of all possible defective sets. Note
that $|\mc{S}_d| = {N \choose K}$.
From (\ref{eq:errevnt_row_col}) in the proof of Proposition~\ref{prop_main_pe_ub}, we note that
  \begin{align} \label{eq:errevnt_unif}
    \mc{E} \subset \left \{ \mathop{\cup}_{S_d \in \mc{S}_d} \mathop{\cup}_{i \in S_d}  \mathop{\cup}_{S_z \in \mc{S}_z} 
    \mathop{\cap}_{j \in S_z} (\mc{E}_{ij} \cap {\mc{H}_i}) \right \} \bigcup 
    \left \{ \mathop{\cup}_{S_d \in \mc{S}_d} \mathop{\cup}_{i \in S_d} \overline{\mc{H}_i} \right \},
  \end{align}

Thus, for the first term in (\ref{eq:main_pe_rbcb_ub}), an additional multiplicative factor of
${N \choose K}$ is needed to account for all possible defective sets.
For the second term, we note that
\begin{align}
  \mathop{\cup}_{S_d \in \mc{S}_d} \mathop{\cup}_{i \in S_d} \overline{\mc{H}_i} \subseteq \mathop{\cup}_{i \in [N]} \overline{\mc{H}_i}. 
\end{align}
Thus, for the second term in (\ref{eq:main_pe_rbcb_ub}), the multiplicative factor of $K$ in (\ref{eq:main_pe_rbcb_ub}) gets replaced by a factor of $N$, and no additional combinatorial multiplicative factors are needed.
The corollary now follows using the same steps as in the proof of {\bf RoAl} and {\bf CoAl}.

\subsection{Proof of Theorem~\ref{thm_A3}} \label{sec_prf_a3}
Let $\matX \in \{0, 1\}^{M \times N}$ denote the random test matrix, 
$\ul{y} \in \{0,1\}^M$ the output of the group test, $Y_z \triangleq \{l \in [M]: \uy(l) = 0 \}$ with
$M_z \triangleq |Y_z|$, and $Y_p \triangleq \{l \in [M]: \uy(l) = 1 \}$ with $M_p \triangleq |Y_p|$.
Let $\matX_z \triangleq \matX(Y_z,:)$ and $\matX_p \triangleq \matX(Y_p,:)$.
Note that $\matX_z \in \{0, 1\}^{M_z \times N}$ and $\matX_p \in \{0, 1\}^{M_p \times N}$.
For the ease of performance analysis of the  LP described in (\ref{eq:row_lp_0}), we work 
with the following equivalent program:
\begin{align} \label{eq:row_lp_0a}
  & &\mathop{\text{minimize}}_{\ul{z}} ~ ~ & ~ ~  \ul{1}_{M_z}^T \matX_z~ \ul{z}  \\
  \nonumber
  & \text{\bf (LP0a)} &\text{subject to} ~ ~ & ~ ~ \ul{0}_{N} \preccurlyeq \ul{z} 
  \preccurlyeq \ul{1}_{N}, ~ ~ ~ ~ ~ ~ ~ ~ ~ ~ ~ ~ ~ ~ ~ ~ ~ ~ ~ ~ ~ ~ ~  \\
  \nonumber
  & & & ~ ~ \ul{1}_{N}^T \ul{z} \ge (N-L).
\end{align}
The above formulation has been arrived at by eliminating the equality constraints 
and replacing the optimization variable $\ul{z}$ by $(\ul{1}_N - \ul{z})$. 
Hence, the non-defective subset output by
(\ref{eq:row_lp_0a}) is indexed by the \emph{smallest} $L$ entries in the solution 
of ({LP0a}) (as opposed to largest $L$ entries in the solution of ({LP0})). 
We know that strong duality holds for a linear program and that any pair
of primal and dual optimal points satisfy the 
Karush-Kuhn-Tucker (KKT) conditions\cite{boydcvxbook}. 
Hence, a characterization of the primal solution can be obtained in terms of the dual 
optimal points by using the KKT conditions.
Let $\ul{\lambda}_1, \ul{\lambda}_2 \in \mathbb{R}^{N}$ and $\nu \in \mathbb{R}$ denote the dual variables
associated with the inequality constraints in (LP0a). 
The KKT conditions for any pair
of primal and dual optimal points corresponding to (LP0a) can be written as follows:
\begin{align} 
  \label{eq:lp0_kkt_grad}
  & \ul{1}_{M_z}^T \matX_z - \ul{\lambda}_1 + \ul{\lambda}_2 - \nu \ul{1}_N = \ul{0}_{N} \\
  \label{eq:lp0_kkt_cs}
  & \ul{\lambda}_1 \circ \ul{z} = \ul{0}_N; ~ \ul{\lambda}_2 \circ (\ul{z} - \ul{1}_N) = \ul{0}_N; ~
  \nu (\ul{1}_N^T \ul{z} - (N-L)) = 0; \\
  \label{eq:lp0_kkt_feas}
  & \ul{0}_{N} \preccurlyeq \ul{z} \preccurlyeq \ul{1}_{N};~
  \ul{1}_{N}^T \ul{z} \ge (N-L);~
  \ul{\lambda}_1 \succcurlyeq \ul{0}_N; ~
  \ul{\lambda}_2 \succcurlyeq \ul{0}_N; ~
  \nu \ge 0; 
\end{align}
Let $(\ul{z}, \ul{\lambda}_1, \ul{\lambda}_2, \nu)$ be the primal, dual optimal
point, i.e., a point satisfying the set of equations (\ref{eq:lp0_kkt_grad})-(\ref{eq:lp0_kkt_feas}).
Let $S_d$ denote the set of defective items. Further, let $\hat{S}_L$ denote the 
index set corresponding to the smallest $L$ entries, and hence the declared 
set of non-defective items, in the primal solution $\ul{z}$. We first derive a sufficient
condition for successful non-defective subset recovery with {\bf RoLpAl}.
\begin{proposition} \label{rblp_prop1}
  If $\ul{\lambda}_2(i) > 0 ~ \forall ~ i \in S_d$, then $\hat{S}_L \cap S_d = \{\emptyset\}$.
\end{proposition}
\noindent Proof: See Appendix \ref{prf_rblp_prop1}.

Let $\mc{E}$, $\mb{P}(\mc{E})$, $S_z$ and $\mc{S}_z$ be as defined at the 
beginning of this section.
The above sufficiency condition for successful non-defective subset recovery,
in turn, leads to the following: 
\begin{proposition} \label{rblp_prop2}
  The error event associated with {\bf RoLpAl} satisfies:
\begin{align} \label{eq:pe0_final}
  \mc{E} \subseteq \mathop{\cup}_{i \in S_d} 
  \mathop{\cup}_{ S_z \in \mc{S}_z} \left \{\ul{1}_{M_z}^T \matX_z(:,i)~  \ge ~
    \ul{1}_{M_z}^T \matX_z(:,j), \forall j \in S_z \right \}.
\end{align}
\end{proposition}
\begin{proof}
  Define $\mc{E}_0(i) \triangleq \{ \ul{\lambda}_2(i) = 0\}$.
We first note, from (\ref{eq:lp0_kkt_grad}), that for any $i \in [N]$
\begin{align} \label{eq:l2eq0implies}
  \ul{\lambda}_2(i) = 0 ~ ~ \Longrightarrow ~ ~\ul{1}_{M_z}^T \matX_z(:,i) 
  = \ul{\lambda}_1(i) + \nu \ge \nu. 
\end{align}
Define $\theta_0 \triangleq \max_{\{i: \ul{\lambda}_1(i) = 0 \}} \ul{1}_{M_z}^T \matX_z(:,i)$
and  $\theta_1 \triangleq \min_{\{i: \ul{\lambda}_1(i) > 0 \}} \ul{1}_{M_z}^T \matX_z(:,i)$.
We relate $\theta_0$, $\theta_1$ and $\nu$ as follows:
\begin{proposition} \label{rblp_prop3} 
  The dual optimal variable $\nu$ satisfies $\theta_0 \le \nu < \theta_1$.
\end{proposition}
\noindent Proof: See Appendix \ref{prf_rblp_prop3}.

From the above proposition and (\ref{eq:l2eq0implies}) it follows that
\begin{align} \label{eq:l2eq0implies_1}
\mc{E}_0(i) \subseteq \left \{ \ul{1}_{M_z}^T \matX_z(:,i) \ge \theta_0 \right \}.
\end{align}
We note that there exists at most $L$ items 
for which $\ul{\lambda}_1(i) > 0$; otherwise the solution would violate the primal feasibility 
constraint:  $\ul{1}_N^T \ul{z}(i) \ge (N-L)$.
Thus, there exist 
  at least $(N-K) - (L-1)$ non-defective items in the set $\{i : \ul{\lambda}_1(i) = 0 \}$. From (\ref{eq:l2eq0implies_1}),  there exists a set $S_z$ of $(N-K)-(L-1)$ non-defective items such that $\left \{ \ul{1}_{M_z}^T \matX_z(:,i) \ge
  \ul{1}_{M_z}^T \matX_z(:,j), \forall j \in S_z \right \}$. Taking the union bound over all possible $S_z$, we get
\begin{align}
  \mc{E}_0(i) &\subseteq \mathop{\cup}_{S_z \in \mc{S}_z} \left \{ \ul{1}_{M_z}^T \matX_z(:,i) \ge
  \ul{1}_{M_z}^T \matX_z(:,j), \forall j \in S_z \right \},
\end{align}
and (\ref{eq:pe0_final}) now follows since using Proposition \ref{rblp_prop1} we have, 
$\mc{E} \subseteq \cup_{i \in S_d} \mc{E}_0(i)$.
\end{proof}
Note that, for a given $i$, the quantity $\ul{1}_{M_z}^T \matX_z(:,i)$ is the same as
the quantity $\mc{T}(i)$ with $\psi_{cb} = 0$ as defined in the proof of Theorem~\ref{thm_A1A2}, and 
(\ref{eq:pe0_final}) is the same as (\ref{eq:row_errevnt}).
Thus, following the same analysis as in Section~\ref{sec_prf_thm_A2}, it follows that,
if $M$ satisfies (\ref{eq:M_cond_noisy_rba1}) with $\psi_0 = 0$, the LP relaxation based algorithm
{\bf RoLpAl} succeeds in recovering $L$ non-defective items with probability exceeding
$1 - \exp\left( -c_0 \log \left [K {N-K \choose L - 1}\right]  \right) 
+ \exp(-c_0 \log K)$.

\subsection{Proof Sketch for Theorem \ref{thm_A4}} \label{sec_prf_thmA4}
We use the same notation as in Theorem \ref{thm_A3} and analyze an
equivalent program that is obtained by replacing $(1-\ul{z})$ by $\ul{z}$.
We note that {\bf LP2} differs from {\bf LP0} only in terms of the objective function, 
and the constraint set remains the same. 
And thus, the complimentary slackness and the primal dual feasibility
conditions are the same as given in (\ref{eq:lp0_kkt_cs}) and (\ref{eq:lp0_kkt_feas}),
respectively. The zero gradient condition for {\bf LP2} is given by:
\begin{align} 
  \label{eq:lp2_kkt_grad}
  \ul{1}_{M_z}^T \matX_z - \psi_{lp} \ul{1}_{M_p}^T \matX_p - \ul{\lambda}_1 + \ul{\lambda}_2 - \nu \ul{1}_N = \ul{0}_{N}.
\end{align}

Let the error event associated with {\bf CoLpAl} be denoted by $\mc{E}$.
Let $i \in S_d$, and define $\mc{E}_i \triangleq \{ i \in \hat{S}_L\}$.
Note that $\mc{E} \subseteq \cup_{i \in S_d} \mc{E}_i$.
Further, it follows that $\mc{E}_i \subseteq \mc{A}_i \cup \mc{B}_i$, where
$\mc{A}_i \triangleq \{ \lambda_2(i) = 0\}$ and 
$\mc{B}_i \triangleq \left \{ \mc{E}_i \cap \{\lambda_2(i) > 0\} \right \}$.
Let us first analyze $\mc{B}_i$. 
Using similar arguments as in Propositions~\ref{rblp_prop1} and \ref{rblp_prop3}, it 
can be shown that,
\begin{align} \label{eq:Bi_upbnd}
  \mc{B}_i \subseteq \{ \nu = 0 \} \subseteq 
  \mathop{\cup}_{ S_z \in \mc{S}_z} \left \{\ul{1}_{M_z}^T \matX_z(:,j) - 
  \psi_{lp} \ul{1}_{M_p}^T \matX_p(:,j) \le 0, \forall j \in S_z \right \},
\end{align}
where $S_z \subset [N]\backslash S_d$ is any set of non-defective items 
such that $|S_z| = (N-K) - (L-1)$ and $\mc{S}_z$ denotes all such sets possible.
Further, using similar arguments as in the proof of 
Theorem \ref{thm_A3}, it can be shown that 
\begin{align} \label{eq:Ai_upbnd}
  \mc{A}_i \subseteq 
  \mathop{\cup}_{ S_z \in \mc{S}_z} \left 
  \{\ul{1}_{M_z}^T \matX_z(:,i) - \psi_{lp} \ul{1}_{M_p}^T \matX_p(:,i) ~  \ge ~
  \ul{1}_{M_z}^T \matX_z(:,j) - \psi_{lp} \ul{1}_{M_p}^T \matX_p(:,j), \forall j \in S_z \right \},
\end{align}
where $S_z$ and $\mc{S}_z$ are as defined above.

The subsequent analysis follows by using the Bernstein inequality to upper 
bound the probability of events $\mc{A}_i$ and $\mc{B}_i$ in a manner similar 
to the previous proofs; we omit the details for the sake of brevity.
Define $\psi_0' \triangleq \min \left( \frac{\gamma_0 \Gamma} {1 - \gamma_0 \Gamma}, \frac{\Gamma}{2 (1 - \Gamma)}\right )$.
Note that, with $\psi_{lp} = \psi_0'$,
$\mb{E}( [\ul{1}_{M_z}^T \matX_z(:,j) - \psi_{lp} \ul{1}_{M_p}^T \matX_p(:,j)]) \ge M p \Gamma/2 > 0$
for any $j \in S_z$. This helps in upper bounding the probability of $\mathcal{B}_i$ using Bernstein's inequality. 
In essence, it can be shown that there exists an absolute constant $C_{4b}>0$, such that
\begin{align} \label{eq:pe_thm4_Bi}
  \mb{P}(\mathop{\cup}_{i \in S_d} \mc{B}_i) \le 
  \exp \lb - \left\{ \frac{M p \Gamma N_0}{C_{4b}} -  
  {\log \left[ K {N-K \choose L -1 } \right ]} \right \} \rb.
\end{align}

Similarly, following the same steps as in the proof of Theorem~\ref{thm_A1A2}, it
can be shown that, for the chosen value of $\psi_{lp}$, there exists an absolute constant
$C_{4a}$ such that
\begin{align} \label{eq:pe_thm4_Ai}
  \mb{P}(\mathop{\cup}_{i \in S_d} \mc{A}_i) \le 
  \exp \lb - \frac{M p \Gamma (1 - \gamma_0)^2 N_0}{C_{4a}} 
  +  {\log \left[ K {N-K \choose L -1 } \right ]} \rb + 
  \exp \lb - \frac{M p \Gamma (1 - \gamma_0)^2 }{C_{4c}} +  {\log K }\rb.
\end{align}
The final result now follows by substituting $p=\frac{1}{(1-u)K}$,
since, by choosing $M$ as in (\ref{eq:M_cond_noisy_rba1}) with $\psi_0 = 0$,
$C_{a1} = \max\{C_{4a}, C_{4b}\}$ and $C_{a2} = C_{4c}$, the
total error in (\ref{eq:pe_thm4_Bi}), (\ref{eq:pe_thm4_Ai}) can be upper bounded
as $2\exp \left( -c_0 \log \lb K {N-K \choose L - 1} \rb \right) + \exp \left( -c_0 \log K \right)$. This concludes the proof.

\section{Simulations} \label{sec:simulations}
In this section, we investigate the empirical performance of the algorithms proposed
in this work for non-defective subset recovery.
In contrast to the previous section, where theoretical guarantees on the number of 
tests were derived based on the analysis of the upper bounds on probability of 
error of these algorithms, here we find the exact number of tests required to achieve a given
performance level, thus highlighting the practical ability of the proposed algorithms
to recover a non-defective subset. This, apart from validating the general theoretical trends, also 
facilitates a direct comparison of the presented algorithms. 

Our setup is as follows. For a given set of operating parameters,
i.e., $N$, $K$, $u$, $q$ and $M$, we choose a defective set $S_d \subset [N]$ randomly such that
$|S_d| = K$ and generate the test output vector $\ul{y}$ according
to (\ref{eq:gtmodel}). We then recover a subset of $L$ non-defective items using the different recovery
algorithms, i.e., {\bf RoAl}, {\bf CoAl}, {\bf RoLpAl}, {\bf RoLpAl++} and {\bf CoLpAl}, and compare it with
the defective set. 
The empirical probability of error is set equal to the fraction of the
trials for which the recovery was not successful, i.e., the output non-defective subset 
contained at least one defective item. 
This experiment is repeated for different values of $M$ and $L$.
For each trial, the test matrix $\matX$ is generated with random  
Bernoulli i.i.d.\ entries, i.e., $\matX_{ij} \sim \mc{B}(p)$,
where $p = 1/K$.
Also, for {\bf CoAl} and {\bf CoLpAl}, we set $\psi_{cb} = \frac{\gamma_0\Gamma }{1 - \gamma_0 \Gamma}$
and $\psi_{lp} = \min \lb \frac{\gamma_0\Gamma }{1 - \gamma_0 \Gamma}, \frac{\Gamma}{2 (1 - \Gamma)} \rb$, 
respectively.
Unless otherwise stated, we set $N=256$, $K=16$, $u=0.05$, $q=0.1$ and we vary
$L$ and $M$.

Figure~\ref{figure:pe_vs_M} shows the variation of the empirical probability of error with the
number of tests, for $L=64$ and $L=128$. 
These curves demonstrate the theoretically expected exponential 
behavior of the average error rates, the similarity of the error rate 
performance of algorithms {\bf RoAl} and {\bf RoLpAl}, 
and the performance improvement offered by {\bf RoLpAl++} at higher values of $L$.
We also note that, as expected, the algorithms that use tests with both 
positive and negative outcomes perform better than the algorithms that use 
only tests with negative outcomes.

Figure~\ref{figure:M_vs_L} presents the number of tests $M$ required to achieve a 
target error rate of $10\%$ as a function of the {size} of the non-defective subset, $L$.
We note that for small values of $L$, the algorithms perform similarly,
but, in general, {\bf CoAl} and {\bf CoLpAl}  are the best performing algorithms
across all values of $L$.
We also note that, as argued in Section~\ref{lin_alg}, {\bf RoLpAl++} performs 
similar to {\bf RoLpAl} for small values of $L$ and for large values of $L$ 
the performance of the former is the same as that of {\bf CoLpAl}.
Also, as mentioned in Section~\ref{sec_alg_disc}, we note the linear increase in
$M$ with $L$, especially for small values of $L$.
We also compare the algorithms proposed in this work with
an algorithm that identifies the non-defective items by first identifying 
the defective items, i.e., we compare the ``direct'' and ``indirect'' 
approach~\cite{AC13} of identifying a non-defective subset.
We first employ a defective set recovery algorithm for
identifying the defective set and then choose $L$ items uniformly 
at random from the complement set. This algorithm is referred to as
``InDirAl'' in Figure~\ref{figure:M_vs_L}.
In particular, we have used ``{\bf No-LiPo-}'' algorithm\cite{jaggi_gtalgo} for 
defective set identification.
It can be easily seen that the ``direct'' approach significantly outperforms the
``indirect'' approach.
We also compare against a non-adaptive 
scheme that tests items one-by-one.  
The item to be tested in each test is chosen uniformly at random from the population.  
We choose the top $L$ items tested in all the tests with negative outcomes
as the non-defective subset. This algorithm is referred to at ``NA1by1'' (Non-Adaptive
$1$-by-$1$) in Figure~\ref{figure:M_vs_L}. It is easy to see that the group testing
based algorithms significantly outperform the NA1by1 strategy.

Figure~\ref{figure:M_vs_L_lbnd_comp} compares the number of tests required
to achieve a target error rate of $10\%$ for {\bf CoAl} with the 
information theoretic lower bound for two different values of $K$.\footnote{We refer the
reader to Theorem $3$ and Section IV in \cite{AC13} for a detailed 
discussion on the information theoretic lower bound.  
Also, see equations (7) and (9) in  \cite{johnson_gtbp} for the derivation of the mutual information 
term that is required for computing the lower bound for the group testing signal model.} 
It can be seen that the 
empirical performance of {\bf CoAl} is within $O(\log K)$ 
of the lower bound. The performance of the other algorithms is found to obey a similar behavior. 


As discussed in Section~\ref{sec_alg_disc}, the parameter settings 
require the knowledge of $K$.
Here, we investigate the sensitivity of the algorithms on the test matrix designed
assuming a nominal value of $K$ to mismatches in its value.
Let the true number of defective items be $K_t$. Let $M(\hat{K}, K_t)$
denote the number of tests required to achieve a given error rate when the test
is designed with $K = \hat{K}$. Let $\Delta_M(\hat{K}, K_t) \triangleq 
\frac{M(\hat{K},K_t)}{M(K_t,K_t)}$. Thus, $\Delta_M(\hat{K}, K_t)$ represents
the penalty paid compared to the case when the test is designed knowing the
number of defective items.
Table~\ref{tab:robust_K} shows the empirically computed $\Delta_M$ for 
different values of uncertainty factor $\Delta_K \triangleq \frac{\hat{K}}{K_t}$ 
for the different algorithms.
We see that the algorithms exhibit robustness to the uncertainty in 
the knowledge of $K$. For example, even when $\hat{K} = 2 K_t$,
i.e., $\Delta_k = 2$, we only pay a penalty of approximately $17\%$ for
most of the algorithms. Also, as suggested by the analysis of the  
upper bounds in Section~\ref{sec_alg_disc}, the 
algorithms exhibit asymmetric behavior in terms of robustness and are more 
robust for $\Delta_k > 1$ compared to when $\Delta_k < 1$.

Figure~\ref{figure:pe_vs_q_u} shows the performance of different algorithms
with the variations in the system noise parameters.
Again, in agreement with the analysis of the probability of error, 
the algorithms perform similarly with respect to variations in both the additive 
and dilution noise.

\begin{figure}[t]
\centering
\includegraphics[scale=0.8]{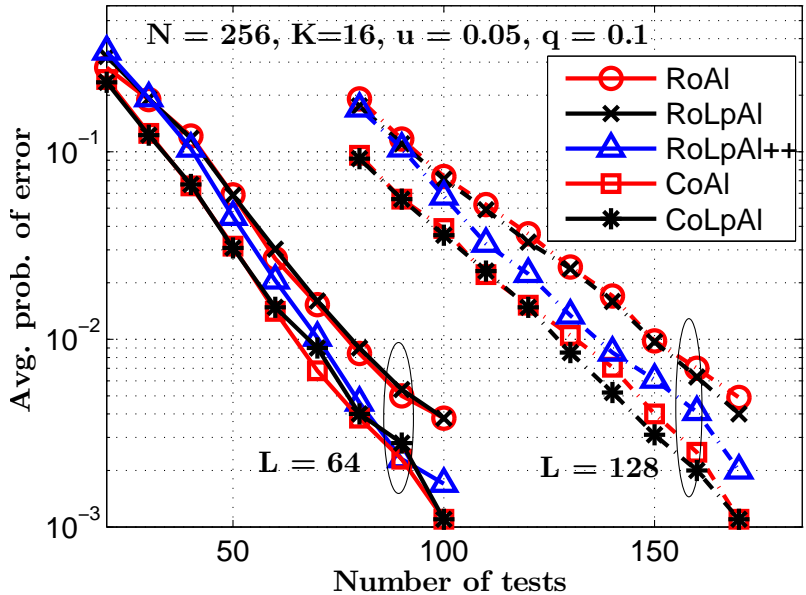}
\caption{Average probability of error (APER) \emph{vs.} number of tests $M$.
The APER decays exponentially with $M$.}
\label{figure:pe_vs_M}
\end{figure}

\begin{figure}[t]
\centering
\includegraphics[scale=0.8]{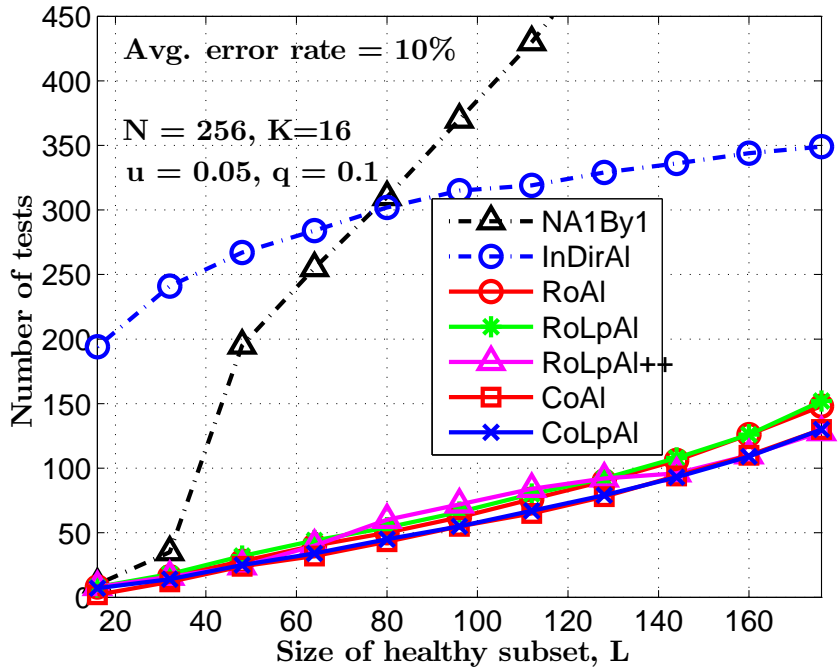}
\caption{Number of tests \emph{vs.} size of non-defective subset. Algorithm {\bf CoLpAl} performs the
best among the ones considered. The direct approach for finding non-defective items significantly 
outperforms both the indirect
approach (``InDirAl''), where defective items are identified first and the non-defective items 
are subsequently chosen from the complement set \cite{AC13}, as well as the item-by-item testing approach (``NA1By1'').}
\label{figure:M_vs_L}
\end{figure}


\begin{figure}[t]
\centering
\includegraphics[scale=0.8]{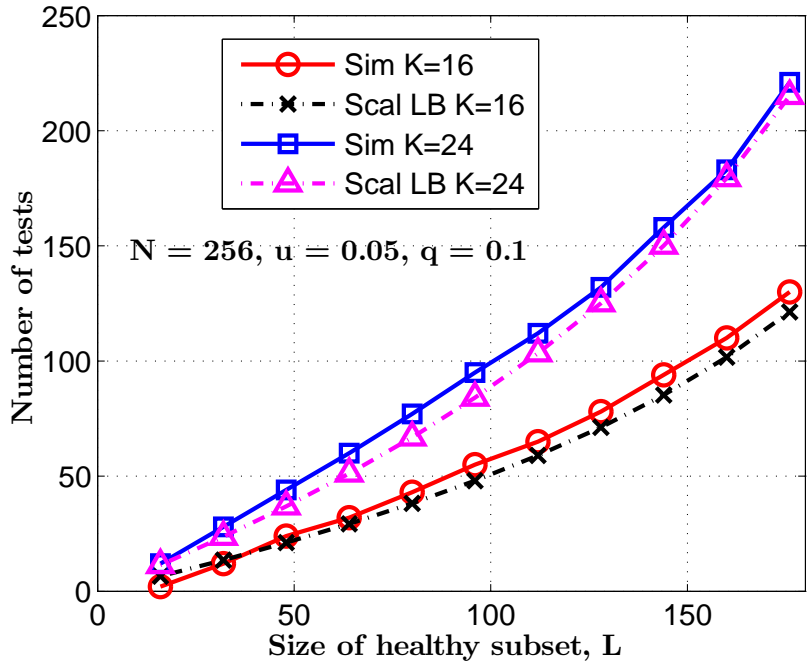}
\caption{Comparison of {\bf CoAl} with the \emph{scaled} information theoretic 
lower bounds. Here, the lower bounds have been scaled by a multiplicative factor of $\log(K)$.
The close agreement of the scaled lower bound with the performance of the algorithm shows that {\bf CoAl} is within a $\log(K)$ factor of the lower bounds.} 
\label{figure:M_vs_L_lbnd_comp}
\end{figure}

%

\begin{figure}[t]
\centering
\subfigure[]{
\includegraphics[scale=0.5]{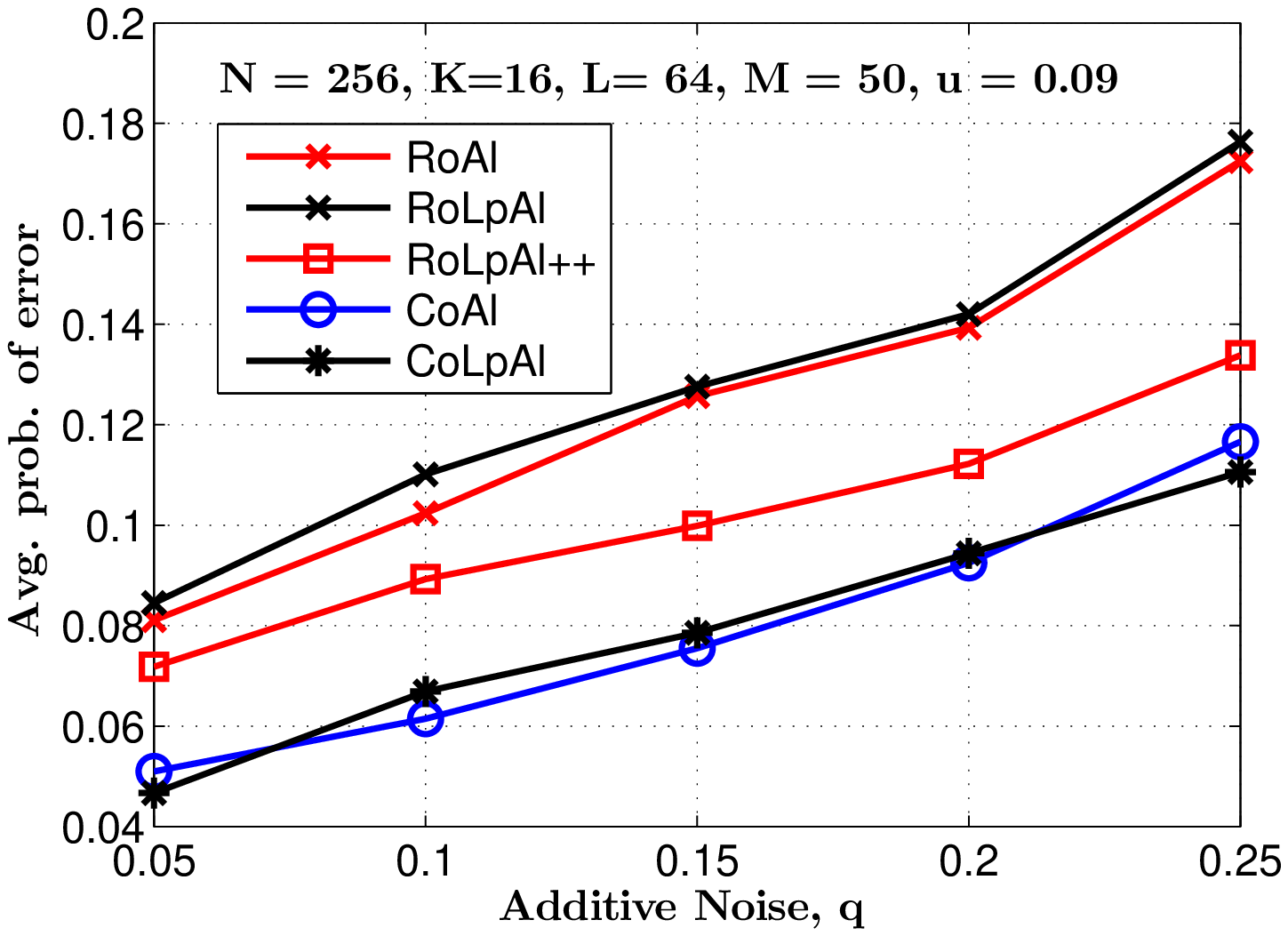}
}
\subfigure[]{
\includegraphics[scale=0.5]{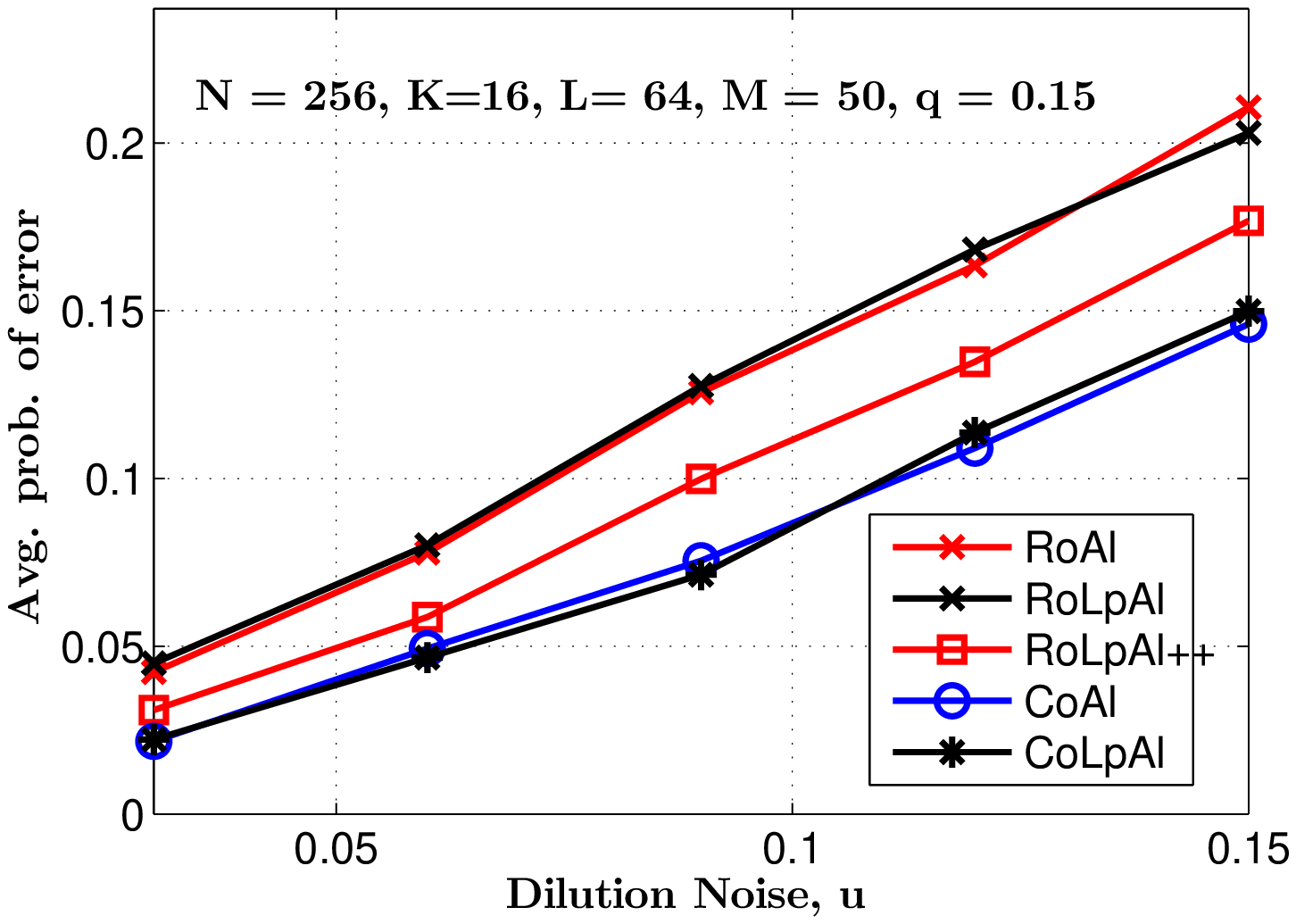}
}
\caption{Variation of the average probability with (a) additive noise ($q$) and (b) dilution noise ($u$).
}
\label{figure:pe_vs_q_u}
\end{figure}


\begin{table}[t]
  \centering
  \caption{Robustness of the non-defective subset identification algorithms to uncertainty in the knowledge of~$K$. The numbers in the
  table are~$\Delta_M(\hat{K}, K_t)$.}
  \begin{tabular}{|l|c|c|c|} \hline 
    \multicolumn{4}{|c|}{$K_t = 16$, $N=256$, $L=128$, $q = 0.1$, $u=0.05$} \\ \hline
    & $\Delta_K = 0.75$ & $\Delta_K = 1.5$ & $\Delta_K = 2.0$  \\ \hline
    {\bf RoAl}  & $1.13$ & $1.06$ & $1.20$  \\ \hline
    {\bf CoAl}  & $1.13$ & $1.04$ & $1.17$  \\ \hline
    {\bf RoLpAl}  & $1.09$ & $1.04$ & $1.17$  \\ \hline
    {\bf RoLpAl++}  & $1.04$ & $1.00$ & $1.17$  \\ \hline
    {\bf CoLpAl}  & $1.11$ & $1.03$ & $1.19$  \\ \hline
  \end{tabular}
  \label{tab:robust_K}
\end{table}

\section{Conclusions} \label{sec_conclusions}

In this work, we have proposed analytically tractable and computationally efficient algorithms 
for identifying a non-defective subset of a given size in a noisy non-adaptive group testing setup.
We have derived upper bounds on the number of tests for guaranteed correct subset 
identification and we have shown that 
the upper bounds and information theoretic lower bounds are order-wise tight
up to a poly-log factor. We have shown that the algorithms are robust to the 
uncertainty in the knowledge of system parameters.
Also, it was found that the algorithms that use both positive and negative outcomes, namely
{\bf CoAl} and the LP relaxation based {\bf CoLpAl},
gave the best performance for a wide range of values of $L$, the size of non-defective
subset to be identified.
In this work, we have considered the randomized pooling strategy. It will be interesting
to study deterministic constructions for the purpose of non-defective subset identification;
this could be considered in a future extension of this work. 
Another interesting question to investigate is to extend the non-defective subset identification problem 
to scenarios with structured pooling strategies, e.g., for graph constrained group 
testing where the pools are constrained by the nodes that lie on a path of a given graph.

\appendix
\subsection{Proof of Lemma \ref{sig_mod_facts}} \label{prf_sig_mod_facts}
We note that a test outcome is $0$ only if none of the $K$ defective items participate
in the test and the output is not corrupted by the additive noise. (a) now follows
by noting that the probability that an item does not participate in the group test
is given by $(1-p) + pu$. (b) follows from (\ref{eq:gtmodel}). For (c) we note that,
given that $X_{li} = 1$ for any $i \in S_d$, the outcome is $0$ only if the $i^{\text{th}}$ item does
not participate in the test (despite $X_{li} = 1$) \emph{and} none of the 
remaining $K-1$ defective items participate in the test (either the entry of the test matrix is zero or the  item gets diluted out by noise) \emph{and} the test outcome is not
corrupted by additive noise. That is, $\mb{P}(Y_l = 0 | X_{li} = 1) = 
u (1 - (1-p)u)^{K-1} (1-q) = \gamma_0 \Gamma$.
The other part follows similarly. 
(d) follows by noting that for any $i \in S_d$ and $j \notin S_d$, 
$\mb{P}(Y_l |  X_{li}, X_{lj}) = \mb{P}(Y_l |  X_{li})$. By Bayes rule and part (b) in 
this lemma, we get:
$\mb{P}(X_{li}, X_{lj} | Y_l) = 
\frac{\mb{P}(Y_l |  X_{li}, X_{lj})}{\mb{P}(Y_l)}  \mb{P}(X_{li}) \mb{P}(X_{lj})
= \mb{P}(X_{li}|Y_l) \mb{P}(X_{lj}|Y_l)$. Hence the proof.

\subsection{Proof of (\ref{eq:mean_ij}) and (\ref{eq:var_j}) } \label{prf_Z0Z1_stats}
For $i \in S_d$ and $j \notin S_d$, let $\ul{x}_j(l) \triangleq X_{jl}$, $\ul{x}_i(l) \triangleq X_{il}$ 
and $\ul{y}(l) \triangleq Y_l$.
For any $k \in [N]$, we note that $\mc{T}(k)$ can be written as a sum of $M$ independent random variables $\sum_{l=1}^M Z_{kl}$, 
where $Z_{kl}$ takes value $1$
with probability $\mb{P}( X_{kl} = 1, Y_l = 0)$,  
$-\psi_{cb}$ with probability $\mb{P}( X_{kl} = 1, Y_l = 1)$, and takes the value $0$ otherwise.
From Lemma~\ref{sig_mod_facts}, we know that
$P(Y_l=0|X_{il} = 1) = \gamma_0 \Gamma$ and $P(Y_l=0|X_{jl} = 1) = \Gamma$ and thus
(\ref{eq:mean_ij}) follows.
Further, (\ref{eq:var_j}) follows by noting that
\begin{align} \nonumber
  \text{Var}(Z_{jl}) \le \mb{E}(Z_{jl}^2) &= p \lb \Gamma + \psi_{cb}^2 (1-\Gamma) \rb\\
\nonumber
\text{Var}(Z_{il}) \le \mb{E}(Z_{il}^2)  &=  p \lb \gamma_0 \Gamma + \psi_{cb}^2 (1- \gamma_0 \Gamma) \rb.
\end{align}
\subsection{Proof of Proposition \ref{rblp_prop1}} \label{prf_rblp_prop1}
We first prove that, for all $i \in \hat{S}_L$, 
$\ul{\lambda}_2(i) = 0$. The proof is based on contradiction. 
Suppose $\exists\ j \in \hat{S}_L$ such that $\ul{\lambda}_2(j) > 0$.
This implies, from the complimentary slackness conditions (\ref{eq:lp0_kkt_cs}), $\ul{z}(j) = 1$ and thus,
$\ul{\lambda}_1(j) = 0$. Since $j^{\text{th}}$ item is amongst the smallest $L$ entries, this implies that 
$\ul{1}_N^T \ul{z} > (N-L)$. Hence, $\nu = 0$. From the zero gradient condition in (\ref{eq:lp0_kkt_grad}), it
follows that $\ul{1}_{M_z}^T \matX_z(:,j) = -\ul{\lambda}_2(j) < 0$, which is not possible, as 
all entries in $\matX$ are nonnegative. It then follows that 
$\forall ~i \in \hat{S}_L$ $\ul{\lambda}_2(i) = 0$. Thus, if 
$\ul{\lambda}_2(i) > 0 ~ \forall ~ i \in S_d$, then these items cannot belong to the first
$L$ entries in the primal solution $\ul{z}$, i.e., $S_d \cap \hat{S}_L = \{ \emptyset \}$.

\subsection{Proof of Proposition \ref{rblp_prop3}} \label{prf_rblp_prop3}
  Suppose $\nu < \theta_0$. Then $\exists~i$ such that $\ul{\lambda}_1(i) = 0$ and 
  $\nu < \ul{1}_{M_z}^T \matX_z(:,i)$. Thus, from (\ref{eq:lp0_kkt_grad}), 
  $\ul{\lambda}_2(i) = \nu - \ul{1}_{M_z}^T \matX_z(:,i) < 0$, which violates the
  dual feasibility conditions (\ref{eq:lp0_kkt_feas}). Thus, $\nu \ge \theta_0$.
  Similarly, let $\nu \ge \theta_1$. Then $\exists~i$ such that $\ul{\lambda}_1(i) = 1$ and 
  $\nu \ge \ul{1}_{M_z}^T \matX_z(:,i)$. Thus, from (\ref{eq:lp0_kkt_grad}), 
  $\ul{\lambda}_2(i) = \ul{\lambda}_1(i) + \nu - \ul{1}_{M_z}^T \matX_z(:,i) \ge 1$, 
  which is a contradiction since $\ul{\lambda}_1(i) > 0$ implies $\ul{\lambda}_2(i) = 0$.
  Thus, $\nu \ge \theta_1$ is not possible.

  \subsection{Affine characterization of the function $\frac{H_b(\alpha)}{1-\alpha}$} \label{app_gamma_approx}
\begin{lemma} \label{lemma_Gamma_approx}
  Let $H_b(\cdot)$ represent the binary entropy function.
  Then, for $0 < \alpha \le \alpha_h < 1$, there exist positive absolute constants 
  $c_0, c_1 > 0$, with $c_1$ depending on $\alpha_h$, such that
  \begin{equation} \label{eq:Gamma_lin}
    \frac{H_b(\alpha)}{1-\alpha} \le c_0 \alpha + c_1.
 \end{equation}
\end{lemma}
  To exablish (\ref{eq:Gamma_lin}), we note that
  \begin{align}
    \nonumber
    \frac{H_b(\alpha)}{1-\alpha} &= -\frac{\alpha}{1-\alpha} \log(\alpha) - \log(1 - \alpha)
    = \frac{\alpha}{1-\alpha}\sum_{i=1}^{\infty} \frac{(1-\alpha)^i}{i} + 
    \sum_{i=1}^{\infty} \frac{\alpha^i}{i}\\
    \nonumber
    &\le \alpha \left(1 + \frac{(1-\alpha)}{2} + \frac{(1-\alpha)^2}{3} \right) + 
    \frac{\alpha (1 - \alpha)^3}{4} \left(\sum_{i=1}^{\infty} (1-\alpha)^{i-1} \right)\\
    \nonumber
    & ~ ~ ~ + \alpha + \frac{\alpha^2}{2} + \frac{\alpha^3}{3} + \frac{\alpha^4}{4} 
    \left(\sum_{i=1}^{\infty} \alpha^{i-1} \right)\\
    \nonumber
    &\le \frac{17}{6} \alpha + \frac{1}{4}\left [(1-\alpha)^3 + \frac{\alpha^4}{1 - \alpha} \right]
    \le c_0 \alpha + c_1,
  \end{align}
  where $c_0 = 17/6$ and $c_1$ is obtained by appropriately bounding the second term when
  $0 \le \alpha \le \alpha_h$.
  In particular, for $\alpha_h \le 0.5$, $c_1 = 0.25$ will satisfy (\ref{eq:Gamma_lin}).

\subsection{Discussion on the theoretical guarantees for {\bf RoLpAl++}} \label{sec_prf_a3pp}
The discussion for {\bf RoLpAl++} proceeds on similar lines as {\bf RoLpAl}.
We use the same notation as in Section \ref{sec_prf_a3}, and, 
as before, we analyze an equivalent LP obtained by 
eliminating the equality constraints and substituting $(1-\ul{z})$ by $\ul{z}$.
The corresponding KKT conditions for a pair of primal and dual optimal points are as follows:
\begin{align} 
  \label{eq:lp1_kkt_grad}
  & \ul{1}_{M_z}^T \matX_z - \ul{\mu}^T \matX_p - \ul{\lambda}_1 + \ul{\lambda}_2 - \nu \ul{1}_N = \ul{0}_{N} \\
  \label{eq:lp1_kkt_cs}
  & \ul{\mu} \circ (\matX_p \ul{z} - (1-\epsilon_0) \ul{1}_{M_p}) = \ul{0}_{M_p}; ~
  \ul{\lambda}_1 \circ \ul{z} = \ul{0}_N; ~ \ul{\lambda}_2 \circ (\ul{z} - \ul{1}_N) = \ul{0}_N; ~
  \nu (\ul{1}_N^T \ul{z} - (N-L)) = 0;  \\
  \label{eq:lp1_kkt_feas}
  & \ul{0}_{N} \preccurlyeq \ul{z} \preccurlyeq \ul{1}_{N};~
  \ul{1}_{N}^T \ul{z} \ge (N-L);~
  \ul{\mu} \succcurlyeq \ul{0}_{M_p}; ~
  \ul{\lambda}_1 \succcurlyeq \ul{0}_N; ~
  \ul{\lambda}_2 \succcurlyeq \ul{0}_N; ~
  \nu \ge 0; 
\end{align}
In the above, $\ul{\mu} \in \mathbb{R}^{M_p}$ is the dual variable associated with constraint 
(\ref{eq:lp1_p_1}) of {\bf LP1}.
Let $(\ul{z}, \ul{\mu}, \ul{\lambda}_1, \ul{\lambda}_2, \nu)$ be a primal, dual optimal point
satisfying the above equations.
We first prove the following: 
\begin{proposition}
  If $\lambda_2(i) > 0$, then $\ul{\mu}^T \matX_p(:, i) = 0$.
\end{proposition}
\begin{proof}
  For any $l \in [M_p]$, if $\matX_p(l, i) = 0$ then $\ul{\mu}(l) \matX_p(l, i) = 0$. 
  If $\matX_p(l, i) = 1$, then for the $l^{\text{th}}$ test 
  $\matX_p(l, :) \ul{z} \ge 1 > (1 - \epsilon_0)$, since $\lambda_2(i) > 0$ implies $z(i) = 1$.
  This implies $\ul{\mu}(l) = 0$, and thus $\ul{\mu}(l) \matX_p(l, i) = 0$. Hence the proposition follows.
\end{proof}

Using the above, it is easy to see that Proposition \ref{rblp_prop1}
holds in this case also. Furthermore, using the same arguments as 
in Section \ref{sec_prf_a3}, it can be shown that the error event associated with
{\bf RoLpAl++}, $\mc{E}$, satisfies 
$\mc{E} \subseteq \mathop{\cup}_{i \in S_d} \mathop{\cup}_{ S_z \in \mc{S}_z} 
\left \{ \cap_{j \in S_z} \mc{E}_{0}(i,j) \right \}$, where
\begin{align} \label{eq:pe_a3pp}
 \mc{E}_0(i, j) = \{\ul{1}_{M_z}^T \matX_z(:,i) - \ul{\mu}^T \matX_p(:,i) -\lambda_1(i)  \ge 
\ul{1}_{M_z}^T \matX_z(:,j) - \ul{\mu}^T \matX_p(:,j) \},
\end{align}
with the notation for $S_z$ and $\mc{S}_z$ as defined at the beginning of Section~\ref{sec_thm_proofs}.
In the following discussion, since $i$ is fixed, for notational simplicity
we will use $\mc{E}_0(j) \triangleq \mc{E}_0(i,j)$.
Note that, for {\bf RoLpAl}, the error event is upper bounded by a similar 
expression as the above but with $\mc{E}_0(j)$ replaced by 
$\mc{E}_1(j) \triangleq \{\ul{1}_{M_z}^T \matX_z(:,i) \ge \ul{1}_{M_z}^T \matX_z(:,j) \}$.
%
In order to analytically compare the performances of
{\bf RoLpAl} and {\bf RoLpAl++}, we try to relate the
events $\mc{E}_0(j)$ and $\mc{E}_1(j)$.
Note that if $\mc{E}_0(j) \subseteq \mc{E}_1(j)$, then $\mb{P}(\mc{E}_0(j)) \le \mb{P}(\mc{E}_1(j))$,
and hence, {\bf RoLpAl++} would outperform {\bf RoLpAl}.
Now, when $\ul{\mu} = \ul{0}_{M_z}$, $\mc{E}_0(j) \subseteq \mc{E}_1(j), ~ \forall j \in S_z$.
For $\ul{\mu} \neq \ul{0}$,
we divide the items in $S_z$ into two disjoint groups:

\noindent (a)  $\lambda_2(j) > 0$: Since $\ul{\mu}^T \matX_p(:, j) = 0$, 
$\ul{\mu}^T \matX_p(:, i) \ge 0$ and $\lambda_1(i) \ge 0$, it follows that
$\mc{E}_0 \subseteq \mc{E}_1$.

\noindent (b) $\lambda_2(j) = 0$: 
We note that $\mc{E}_0(j) \subseteq \mc{E}_1(j) \cup \mc{E}_1'(j)$ where 
$\mc{E}_1'(j) = \left \{ \ul{\mu}^T \left [ 
\matX_p(:,j) - \matX_p(:,i) \right] \ge \kappa + \lambda_1(i) \right \}$,
where $\kappa + \lambda_1(i) > 0$. 

A technical problem, which does not allows us to state the
categorical performance result, arises now.
It is difficult to obtain the estimates for the dual variables $\ul{\mu}$ and hence 
of $\mb{P}(\mc{E}_1'(j))$. Therefore, we offer 
two intuitive arguments that provide insight into the relative 
performance of {\bf RoLpAl++} and {\bf RoLpAl}.
The first argument is that the \emph{majority} of the items 
in $S_z$ will have $\lambda_2(j) > 0$ and 
thus, for a majority items in $S_z$, it follows that 
$\mb{P}(\mc{E}_0(j)) \le \mb{P}(\mc{E}_1(j))$.
This is because the set $\{ j: \lambda_2(j) = 0 \}$ is given by,
\begin{align}
  \left \{ j : 
    \left ( \ul{1}_{M_z}^T \matX_z(:,j) - \ul{\mu}^T \matX_p(:,j) \right )
    = \mathop{\max}_{\{l: \lambda_1(l) = 0 \}} 
    \left ( \ul{1}_{M_z}^T \matX_z(:,l) - \ul{\mu}^T \matX_p(:,l) \right ) \right \},
\end{align}
and, as the number of tests increase and the number of nonzero components of $\ul{\mu}$ 
increase, the probability that above equality holds becomes smaller and smaller.
Furthermore, for a small number of items $j \in S_z$ with $\lambda_2(j) = 0$, it is 
reasonable to expect that $\mb{P}(\mc{E}_1'(j))$ will be small. 
This is because the probability
that a defective item is tested in a pool with positive outcome is higher that
the probability that a non-defective item is tested in a pool with positive outcome. 
Thus, the expected value of $\ul{\mu}^T [\matX(:,j) - \matX(:,i)]$ will be negative for a 
non-negative $\ul{\mu}$
and, thus using concentration of measure arguments, we can expect $\mb{P}(\mc{E}_1'(j))$ 
to be small.
Thus, we expect that {\bf RoLpAl++} to perform similar (or even better) than {\bf RoLpAl}.

\subsection{Chernoff Bounds} \label{sec_chernoff}

\begin{theorem}
  (Bernstein Inequality \cite{Lugosi06com}) Let $X_1, X_2, \ldots, X_n$ be independent 
  real valued random variables, and assume that $|X_i| < c$ with probability one.
  Let $X = \sum_{i=1}^n X_i$, $\mu = \mb{E}(X)$ and $\sigma = \text{Var}(X)$. 
  Then, for any $\delta > 0$, the following hold:
  \begin{align}
    &\mb{P} \lb X > \mu + \delta \rb \le 
    \exp\lb - \frac{\delta^2}{2 \sigma^2 + \frac{2}{3}c \delta } \rb\\
    &\mb{P} \lb X < \mu - \delta \rb \le 
    \exp\lb - \frac{\delta^2}{2 \sigma^2 + \frac{2}{3}c \delta } \rb
  \end{align}
\end{theorem}

\section*{Acknowledgment}
The authors thank Anuva Kulkarni for several interesting discussions in the initial phase of this work.

\bibliographystyle{IEEEbib}
\bibliography{IEEEabrv,bibJournalList,refs_alg_fnds}

\end{document}